\documentclass[11pt,letterpaper]{llncs}

\usepackage{fullpage}
\usepackage[numbers,sort]{natbib}

\usepackage{comment}
\usepackage{graphicx}
\usepackage{amssymb}

\newcommand{\bigO}{O}

\newcommand{\seed}{\bar s}
\newcommand{\rk}{\phi_q}
\newcommand{\rks}{\phi_{q,\seed}}
\newcommand{\rkspow}{\varphi_{q,\seed}}
\newcommand{\modd}{\ \mathrm{mod}\ }
\newcommand{\divv}{\ \mathrm{div}\ }
\newcommand{\packedtime}{t_{n,\sigma}}

\begin{document}

\title{Optimal Substring Equality Queries with Applications to Sparse Text Indexing}  
%
%
\author{Nicola Prezza
\orcidID{0000-0003-3553-4953}} 
%
\authorrunning{N. Prezza}
%
\institute{Luiss Guido Carli, Rome, Italy
\email{nprezza@luiss.it}}
\maketitle              

\begin{abstract}
	
	We consider the problem of encoding a string of length $n$ from an integer alphabet of size $\sigma$ so that \emph{access} and \emph{substring equality} queries (that is, determining the equality of any two substrings) can be answered efficiently. 
	Any uniquely-decodable encoding supporting \emph{access} must take $n\log\sigma + \Theta(\log (n\log\sigma))$ bits. We describe a new data structure matching this lower bound when $\sigma\leq n^{\bigO(1)}$ while supporting both queries in optimal $\bigO(1)$ time. Furthermore, we show that the string can be overwritten in-place with this structure.
	The redundancy of $\Theta(\log n)$ bits and the constant query time break \emph{exponentially} a lower bound that is known to hold in the read-only model.

	Using our new string representation, we obtain the first in-place subquadratic (indeed, even sublinear in some cases) algorithms for several string-processing problems in the \emph{restore} model: the input string is rewritable and must be restored before the computation terminates. In particular, we describe the first in-place subquadratic Monte Carlo solutions to the sparse suffix sorting, sparse LCP array construction, and suffix selection problems. 
	With the sole exception of suffix selection, our algorithms are also 
	the first running in \emph{sublinear} time for small enough sets of input suffixes. 
	Combining these solutions, we obtain the first sublinear-time Monte Carlo algorithm for building the sparse suffix tree in compact space.
	We also show how to derandomize our algorithms using small space. 
	This leads to the first Las Vegas in-place algorithm computing the full LCP array in $\bigO(n\log n)$ time and to the first Las Vegas in-place algorithms solving the sparse suffix sorting and sparse LCP array construction problems in $O(n^{1.5}\sqrt{\log \sigma})$ time. Running times of these Las Vegas algorithms hold in the worst case with high probability.
\end{abstract}

\keywords{Substring-equality queries, In-place, Suffix sorting}

\newpage

\section{Introduction and related work}

In this paper we consider the problem of finding an encoding for a string $S$ efficiently supporting the following queries:

\begin{itemize}
	\item \emph{Random access}: return any $S[i]$.
	\item \emph{Substring-equality}: test whether any two substrings $S[i,\dots,i+\ell]$ and $S[j,\dots,j+\ell]$ are equal. 
\end{itemize}

It is well-known (see for example \cite[Thm. A]{bentley1976almost}) that any uniquely-decodable code for an integer $X$ must use at least $\log X + \Theta(\log\log X)$ bits of space. By viewing a string $S\in [0,\sigma-1]^n$ as a number of $n$ digits in base $\sigma$, we conclude that any uniquely-decodable encoding for $S$ must use at least $\log(\sigma^n) + \Theta(\log\log(\sigma^n)) =  n\log\sigma + \Theta(\log(n\log\sigma))$ bits in the worst case. Given that encodings of this size do exist (for example, Elias codes), the next interesting question is: what kind of queries can we answer efficiently within this space? As it turns out, this lower bound is not achievable for all possible queries. Partial sums on a bitvector of length $n$ (also known as \emph{rank}), for example, require $n + n/\log^{O(t)}n$ bits of space to be solved within $O(t)$ time~\cite{puatracscu2010cell}. Other known data structure trade-offs include evaluating boolean polynomials~\cite{gal2003cell}, text searching, evaluating permutations, and select on binary matrices~\cite{golynski2009cell}. 
On the positive side, Dodis et al.~\cite{dodis2010changing} showed that this is not the case for the fundamental random access and update queries, that is, retrieving and modifying characters of $S$: they provided an optimal encoding answering these queries in $O(1)$ time on an arbitrary finite alphabet. Their structure requires access to $O(\log n)$ word constants (not depending on $S$) of $\Theta(\log n)$ bits each.

While the above lower bounds hold in the encoding setting --- where $S$ is replaced by a data structure --- , an alternative setting is the one where we augment a read-only string with some auxiliary data structures that can access it. 
In this model, an interesting result relevant to our paper is constituted by longest common extension (LCE) queries: to find the longest common prefix between any two string's suffixes.
For this problem, Kosolobov~\cite{kosolobov2017tight} showed that, letting $s(n)$ be the space redundancy (in bits) on top of $S$ and $t(n)$ be the query time, the relation $s(n)t(n) \in \Omega(n\log n)$ must hold. 
This lower bound is matched by a data structure by Bille et al.~\cite{bille2014time,bille2015longest,bille2013sparse}, who showed how to augment a read-only string over alphabet of size $n^{O(1)}$ with a data structure of $\bigO(n/\tau)$ words of $\Theta(\log n)$ bits supporting LCE queries in $\bigO(\tau)$ time.  
It is worth to note that Kosolobov's lower bound holds in the cell-probe model where each character occupies a separate cell and with alphabets of size at least $2^{\Theta(s(n)/n)}$. 
Indeed, packed access can break the lower bound: as Kempa and Kociumaka recently showed in~\cite{Kempa19}, $s(n)t(n) \in O(n\log\sigma)$ can be achieved in the word-RAM model with words of $\Theta(\log n)$ bits and alphabet of size $n^{O(1)}$.
Another recent result breaking the lower bound is represented by the structure of Tanimura et al.\cite{tanimura17}, which for $\sigma \leq 2^{o(\log n)}$ achieves $s(n)t(n) = o(n\log n)$. 
We note that, since substring equality queries can be used to solve LCE queries by binary search, Kosolobov's result implies the lower bound $s(n)t(n) \in \Omega(n)$ for substring equality queries in the read-only model. 
In particular, in order to achieve $O(1)$ words ($O(\log n)$ bits) of redundancy, $\Omega(n/\log n)$ query time is needed (again, packed computation might break this bound by a logarithmic factor). 
One may wonder whether Kosolobov's lower bound holds also in the \emph{restore} model where the string is allowed to be overwritten, provided that it is restored before the computation terminates. This model has been considered by Fischer et al.~\cite{fischer2016deterministic}, who managed to improve the running time and working space used by Bille et al. to build their data structure (although offering worse query times).
The main result of our paper is to show that in the restore model the situation is indeed much better. 
In particular, we show that in this setting $O(1)$ words of redundancy and $O(1)$ query time are achievable for substring equality queries at the price of using $O(\log n)$ additional word constants of $\Theta(\log n)$ bits that do not depend on $S$. 
Even counting these $O(\log n)$ word constants in the redundancy $s(n)$, our result implies that in the restore model $s(n)t(n) \in O(\log^2 n)$ is achievable. 
If we do not count those constants (which can be shared by all instances of our data structure), we obtain $s(n)t(n) \in O(\log n)$. This improves \emph{exponentially} the lower-bound holding in the read-only model. 
Even better, we show that the string can be replaced (and restored) \emph{in-place} and in optimal packed time with our data structure; for this task, we provide both Monte Carlo and Las Vegas algorithms (the latter use small additional working space). 
As we discuss in the next subsection, these results can be used as a powerful tool to solve in-place and subquadratic (in some cases, even sublinear) time many string processing problems for which no subquadratic solution was previously known.

\subsection{Applications To Sparse Text Indexing}

The first application we consider is in-place \emph{suffix sorting}: to compute the lexicographic order of the string's suffixes.
This is a natural application of our result, since the lexicographic order of two suffixes can be determined by comparing the two characters following their longest common prefix. 
Suffix sorting has been the subject of study of dozens of research articles since the introduction of suffix arrays in~\cite{manber1993suffix,baeza1992new,gonnet1992new}, and is a fundamental step in most of the indexing and compression algorithms developed to date.
The survey~\cite{puglisi2007taxonomy} represents an excellent overview of the main suffix sorting techniques developed in the two decades following the introduction of suffix arrays. 

Relevant to our work are the results of Franceschini and Muthukrishnan~\cite{franceschini2007place} and Li et al. \cite{li2018optimal}. 
The former were the first to show that suffix sorting can be performed in-place in optimal $O(n\log n)$ time on general ordered alphabets. They assume a read-only string and a general  rewritable array $[0,n-1]$ to store the suffix array (i.e., not necessarily encoded as integers of $\lceil\log n\rceil$ bits, even if they consider also this model). In the same model, Li et al. \cite{li2018optimal} improved the running time to $O(n)$ in the case the alphabet's size is $\sigma\in O(n)$. 
Parallel to the study of techniques to sort all suffixes of a string, several authors started considering the problem of efficiently sorting only a \emph{subset} of $b$ string's suffixes~\cite{karkkainen1996sparse, karkkainen2014faster,gawrychowski2017sparse,fischer2016deterministic,bille2013sparse,bille2015longest,karkkainen2006linear}, a fundamental step in the construction of compressed and sparse text indexes~\cite{karkkainen1996sparse} and space-efficient compression algorithms. Among these results, Gawrychowski and Kociumaka~\cite{gawrychowski2017sparse} have been the first to describe an asymptotically optimal solution to the problem, showing that $\bigO(b)$ working space and $\bigO(n)$ running time are achievable on alphabet size $n^{O(1)}$ with a Monte Carlo algorithm. 
They assume a read-only input string and also discuss a Las Vegas algorithm with higher running time. More recently, Birenzwige et al.~\cite{Birenzwige20} presented a Las Vegas algorithm running in optimal $O(n)$ time on read-only strings from an alphabet of size $n^{O(1)}$. To the best of our knowledge, Fischer et al.~\cite{fischer2016deterministic} have been the first (excluding the conference verion~\cite{prezzaSparse} of the present paper) to consider the problem in the restore model. On general alphabet size, their solution uses $O(b)$ working space and can achieve sublinear time on some instances  (in particular, when the number of comparisons needed to sort the suffixes is small). Note that, since all these results can use $O(b)$ words of working space, they do not need to make strong assumptions on the output's representation. In particular, they can use any array representation supporting random access and update operations (this situation changes in the in-place setting, read below).

Interestingly, to date no in-place (i.e. $\bigO(1)$ working space) and $o(n\cdot b)$-time algorithm is known for the general sparse suffix sorting problem. Such an algorithm should take as input a string $S$ and an array $B$ of $b$ positions, and suffix-sort these position using $\bigO(1)$ words of working space on top of $S$ and $B$. Using our new data structures, we fill this gap. All our results hold in the restore model where $S$ and $B$ are stored using any integer representation supporting random access and update operations and $S \in [0,\sigma]^n$, with $\sigma\leq n^{O(1)}$, must be restored before the computation terminates. In this model, we present a Monte Carlo in-place algorithm running in $O(n/\log_\sigma n + b\log^2n)$ worst-case time (assuming packed computation is available; we also provide results in a more general setting). 
For a small enough set of input suffixes, this running time is sublinear.
We also describe a Las Vegas in-place algorithm for the problem running in $O(n^{1.5}\sqrt{\log\sigma})$ time with high probability.

The second problem we consider is that of building in-place the (sparse) Longest Common Prefix array (LCP), that is, the array storing the lengths of the longest common prefixes between lexicographically adjacent suffixes (or a subset of them).
Like suffix-array construction algorithms, LCP construction algorithms have been the subject of several research articles in the last decades (see~\cite{puglisi2008space, gog2011fast} and references therein). As opposed to suffix arrays, in-place algorithms for building the LCP array have been considered only very recently~\cite{LOUZA201714}. The time-gap between solutions for building in-place the (full) suffix array and the LCP is considerable, as the fastest known algorithm for the latter problem, due to Louza et al. runs in $\bigO(n^2)$ time~\cite{LOUZA201714}. This algorithm does not make assumptions on the alphabet and works in the rewritable model. The same bound in the read-only model can be obtained by computing the suffix array in-place with the algorithms of Franceschini and Muthukrishnan~\cite{franceschini2007place} and Li et al. \cite{li2018optimal} and replacing it with the LCP array.
Our major contribution to this problem is a Las Vegas in-place algorithm computing the full LCP array in $O(n\log n)$ time with high probability. 
As far as the sparse LCP array is concerned, all solutions previously discussed that solve LCE queries or that build the sparse suffix tree can be adapted to return the sparse LCP array within the same time-space bounds. In particular, no in-place solution working in $o(n\cdot b)$ time is known. For this problem, we describe a Monte Carlo in-place algorithm running in $O(n/\log_\sigma n + b\log^2n)$ worst-case time and a Las Vegas in-place algorithm running in $O(n^{1.5}\sqrt{\log\sigma})$ time with high probability. 
We remark that our results allow using any sequence representation for the input (provided it supports constant-time access and update queries). In particular, using the representation of Dodis et al.~\cite{dodis2010changing} for both $S$ and $B$, our algorithms use $n\log\sigma + b\log n + \Theta(\log n)$ bits of \emph{total} space (that is, including input, output, and working space). 

To conclude, we consider the \emph{suffix selection} problem: to return the $i$-th lexicographically smallest suffix. In the read-only model, it is known that this problem can be solved in optimal $\bigO(n)$ time and in $\bigO(n)$ words of space on top of the string~\cite{franceschini2007optimal} drawn from an arbitrary ordered alphabet. Considering that the output consists of only one position, this solution is far from being space-efficient. We present the first Monte Carlo in-place solution for the suffix selection problem in the restore model that runs in sub-quadratic time. On alphabets of size $n^{O(1)}$, our algorithm runs in $O(n\log^4 n)$ worst-case time and returns the correct result with high probability.

The work described in this paper is an extension of the data structure presented in~\cite{prezzaSparse} for supporting substring equality queries within $n\lceil \log\sigma \rceil + \Theta(\log n)$ bits of space. We generalize the space-optimality of this result to alphabet sizes not being necessarily a power of two. As an immediate result of this generalization, we can abstract from the input representation and present our result as a string transformation. 
With respect to~\cite{prezzaSparse}, we also improve the running times of our in-place algorithms to sublinear for small enough input suffix sets and present efficient Las Vegas algorithms in the sparse setting. 

\subsection{Notation and Model of Computation}\label{sec:notation}

We assume that the input string $S \in [0,\sigma-1]$, with $\sigma \leq n^{O(1)}$, is rewritable and that the size of the data structure used to represent it does not depend on its content, that is, updating the characters of $S$ does not change the space of the data structure. For example, a word-packed string or the representation described in~\cite{dodis2010changing} satisfy this requirement. 
Intuitively, we need this requirement since in our solutions we overwrite the content of $S$ and, in order for this process to run in-place, the size of the string does not have to increase.
All our algorithms work in the \emph{restore} model: the string $S$ is allowed to be overwritten and must be restored before the computation terminates. 
We describe our results in the word RAM model of computation with word size $w = \lceil \log n \rceil$ bits; as we show below, this is not restrictive.
If not otherwise specified, logarithms are in base $2$.
Since we make use only of integer additions, multiplications, modulo, and bitwise operations (masks, shifts), we assume that we can simulate a memory word of size $w'=c\cdot w$ for any constant $c$ with only a constant slowdown in the execution of these operations. Subtractions, additions and multiplications between $(c\cdot w)$-bits  words take trivially constant time by breaking the operands in $2c$ digits of $w/2$ bits each and use schoolbook's algorithms (i.e. $\bigO(c)$-time addition and $\bigO(c^2)$-time multiplication). The modulo operator $a\modd q$ (with $q$ fixed) can be computed as $a\modd q = a - \lfloor a/q \rfloor\cdot q$. Similarly, computing $\lfloor a/q \rfloor$ can be done in $\bigO(c^2)$ time using Knuth's long division algorithm~\cite{knuthbook} (see also~\cite{Hansen1994}).
Bitwise operations on $(c\cdot w)$-bits  words can easily be implemented with $c$ bitwise operations between $w$-bits words.
We enumerate string positions starting from 0.
For a more compact notation, we write $S[i,j]$ to denote both $S$'s substring starting at position $i$ and ending at position $j$, and the integer $S[i]S[i+1]\dots S[j]$, each $S[i]$ being a digit in base $\sigma$. If $j<i$, $S[i,j]$ denotes the empty string $\epsilon$ or the integer $0$. The use (string/integer) will be clear from the context.
We assume that we can draw integers uniformly distributed in a given interval in constant time and using $\bigO(1)$ words of working space.
W.h.p. (with high probability) means with probability at least $1 - n^{-c}$ for an arbitrarily large constant $c$.
The Karp-Rabin hash function~\cite{karp1987efficient} takes as input a string $S\in [0,\sigma-1]^*$, converts it into its natural representation as a number $s$ composed of $|S|$ digits in base $\sigma$ (i.e. the $i$-th most significant digit of $s$ is equal to $S[i]$), and returns $s\modd q$, where $q$ is a prime number. Since in this paper we view strings from $[0,\sigma-1]^*$ as numbers in base $\sigma$, we can define the Karp-Rabin fingerprint $\rk(S)$ of a string $S$ simply as
$$
\rk(S) = S\modd q.
$$

Let now $\seed\in \mathbb{N}$ be a \emph{seed}. With $\rks$ we indicate the \emph{shifted Karp-Rabin function}
$$
\rks(S) = (S+\seed)\modd q = (\rk(S) +\seed)\modd q.
$$

\section{In-Place Substring Equality Queries}\label{sec:encoding}

Let $S[0,n-1]$ be our input string. The intuitive idea of our strategy is to create a string $S'[0,n-1]$ on the same alphabet as follows. Let $\tau \in \Theta(\log_\sigma n)$ be a block size, and $q$ be a random prime of $\tau+1$ digits in base $\sigma$. We conceptually divide $S$ and $S'$ in contiguous non-overlapping blocks of $\tau$ digits each (assume for simplicity that $\tau$ divides $n$). The key insight is to choose $q$ so that all Karp-Rabin fingerprints of the $S$-prefixes ending at block boundaries have only $\tau$ digits. At this point, each block $B'$ of $S'$ contains the $\tau$ digits of the Karp-Rabin fingerprint of the corresponding $S$-prefix. Since we turn global information (prefixes) of $S$ into local information (characters) of $S'$, we are able to answer substring equality queries by accessing only $\tau \in \Theta(\log_\sigma n)$ contiguous locations of $S'$ (which translates to $\bigO(1)$ accesses when $\Theta(\log_\sigma n)$ characters are packed in a single word). By using well-established techniques, we can moreover choose $q$ so that no $S$-substrings used to solve substring equality queries generate collisions through our hash function with high probability. 

In Definition \ref{def:LCE transform} we formally introduce a family of string transforms. In the next subsections we show how to choose a transform from this family with the following properties: 
\begin{itemize}
	\item \emph{Invertible}: we can reconstruct $S$ from its transform and the corresponding parameters.
	\item \emph{Monte Carlo}: the transformed string is invertible \emph{and} with high probability it can be used to answer substring equality queries on $S$, \emph{without accessing $S$}.
\end{itemize}

We will moreover introduce a data structure corresponding to our Monte Carlo transform and show how to derandomize it with a Las Vegas construction algorithm. 


Let $\tau, q, \seed$ be three non-negative integers satisfying the following conditions: $\tau = \Theta(\log_\sigma n)$, $\sigma^{\tau} \leq q < \sigma^{\tau+1}$, and $\seed \leq q$. 
These three integers will have the following intuitive role: 

\begin{itemize}
	\item $\tau$ is a block size. We will divide $S$ (as well as the transformed string $S'$) in contiguous non-overlapping blocks of size $\tau$.
	\item $q$ and $\seed$ are the two parameters of a shifted Karp-Rabin hash function (see Section \ref{sec:notation}). By randomly choosing $q$ and $\seed$ according to an appropriate distribution, we will make sure that the transformed string $S'$ is invertible and that the hash function is collision-free with high probability.
\end{itemize}

Let  $P_j = \rks(S[0,j\cdot \tau-1])$, for $j = 1,\dots, n/\tau$, be $n/\tau$ integers composed of $\tau+1$ digits in base $\sigma$ each (if less, left-pad with zeros), and denote with $P_j[k]$, $0\leq k \leq \tau$, the $k$-th most significant digit of $P_j$.
For simplicity, we assume that $\tau$ divides $n$. 
This will not affect the space of our final encoding:
if $\tau$ does not divide $n$, then we virtually pad the string with $n\modd \tau$ digits equal to zero. For our choice of $\tau$, this results in a waste of $\bigO(\log_\sigma n)$ characters, which can be stored in packed form as a number in base $\sigma$ in $\bigO(1)$ words while still supporting packed access and update queries with a constant number of accesses to memory.
In Definition \ref{def:LCE transform} we introduce our family of string transforms. For brevity, we call such functions \emph{eq-transforms}.

\begin{definition}[eq-transform]\label{def:LCE transform}
	The eq-transform $\lambda_{q, \seed, \tau}(S)\in [0,\sigma-1]^n$ of $S$ is defined as follows. For every $i=0,\dots, n-1$:
	$$
	\lambda_{q, \seed, \tau}(S)[i] = P_{(i\divv\tau)+1}[(i\modd \tau) + 1]
	$$
\end{definition}

Note that, in Definition \ref{def:LCE transform}, we discard all digits $P_j[0]$, for  $j = 1,\dots, n/\tau$.
In order to being able to reconstruct the original string $S$ from its transform $\lambda_{q, \seed, \tau}(S)$, we need to make sure that this omission does not result in a loss of information. We will achieve this goal by choosing $q$ and $\seed$ in such a way that, with high probability, all the discarded digits are equal to $0$.

\subsection{Invertible Eq-Transforms}

In the rest of the paper we focus our attention on \emph{invertible} eq-transforms:

\begin{definition}
	We say that $\lambda_{q, \seed, \tau}(S)$ \emph{is invertible} if and only if we can reconstruct $S$ from $\lambda_{q, \seed, \tau}(S)$, $q$, $\seed$, and $\tau$.
\end{definition}

Recall the definition $P_j = \rks(S[0,j\cdot \tau-1])$ used to define our transform. It is easy to see the following:

\begin{lemma}\label{lem:invertible condition}
	If $P_j <\sigma^\tau$ for each $j=1, \dots, n/\tau$, then
	$\lambda_{q, \seed, \tau}(S)$ is invertible.
\end{lemma}
\begin{proof}
	To extract the $j$-th block $S[(j-1)\cdot \tau,j\cdot \tau-1]$, we first extract $P_j$ and $P_{j-1}$ from $\lambda_{q, \seed, \tau}(S)$ (if $j=1$, take $P_{j-1} = \seed$): since $P_j,P_{j-1} < \sigma^\tau$, then those two values are equal to the $j$-th and $(j-1)$-th blocks of the transform. Then, it is not hard to see that the following equality holds:
	$$
	S[(j-1)\cdot \tau,j\cdot \tau-1] = (P_{j}-\seed) - (P_{j-1}-\seed)\cdot \sigma^\tau \modd q.
	$$
	This implies that we can access any character of $S$ from its transform, that is, the transform can be inverted.
\end{proof}

The goal of the next lemma is to show that --- for an appropriate choice of the parameters $\tau, q, \seed$ and with inverse polynomial failure probability --- we can find an invertible eq-transform for any input string. In the next subsections we will refine this result in order to obtain a \emph{Monte Carlo} transform and a Las Vegas algorithm to derandomize the corresponding data structure.
Let $\tau > 0$ be an integer block size, $h \geq 2$, and assume $n\geq 2$. We define the set of integers $\mathcal Z_{\tau,h}$ as follows.

\begin{equation}\label{eq:Z}
\mathcal Z_{\tau,h} = \left[ \sigma^{\tau}, \left\lfloor \sigma^{\tau}\left( \frac{n^{h-1}}{n^{h-1}-1} \right)\right\rfloor \right].
\end{equation}


\begin{lemma}\label{lem:invertible}
	Given an input string $S\in \Sigma^n$, let:
	\begin{itemize}
		\item $h > 0$ be an arbitrarily large constant,
		\item $\tau > 0$ be any integer block size,
		\item $q$ be any integer chosen from $\mathcal Z_{\tau,h+2}$, and
		\item $\seed$ be a uniform integer chosen from $[0,q-1]$.
	\end{itemize}
	Then, $\lambda_{q, \seed, \tau}(S)$ is an invertible eq-transform with probability at least $1-n^{-h}$.
\end{lemma}	
\begin{proof}
	
	Let $P_i = \rks(S[0,i\cdot \tau-1])$, for $i=1, \dots, n/\tau$ be the shifted Karp-Rabin fingerprint of the prefix ending at the $i$-th block of $S$.	
	Let $x$ be any integer in $[0,q-1]$. Note that $P_i$ is a uniform random variable in the interval $[0,q-1]$. Let $0\leq x < q$:
	\begin{equation}\label{eq:unif}
	\begin{array}{lll}
	\mathcal P(P_i=x) &=& \mathcal P(\rk(S[0,i\cdot \tau-1]) + \seed \equiv_q x) \\
	&=&  \mathcal P(\seed \equiv_q x - \rk(S[0,i\cdot \tau-1])) \\
	&=& 1/q 
	\end{array}
	\end{equation}
	
	Let $h'=h+2$ and recall that $q$ is extracted from $\mathcal Z_{\tau,h'}$. Equation \ref{eq:unif} and the bound $q\leq \sigma^\tau\cdot \frac{n^{h'-1}}{n^{h'-1}-1}$ of Equation \ref{eq:Z} imply that:
	$$
	\mathcal P(P_i \geq \sigma^\tau) = 1-\sigma^\tau/q \leq 1/n^{h'-1} = 1/n^{h+1}
	$$
	
	Note that $\mathcal P(P_i \geq \sigma^\tau)$ is the probability that $P_i$ has more than $\tau$ digits. Let $\mathcal X$ be the random variable denoting the number of blocks $P_i$ having more than $\tau$ digits. The expected value of $\mathcal X$ is
	
	\begin{equation}\label{eq:exp X}
	E[\mathcal X] = \frac{n}{\tau}\cdot  \mathcal P(P_i \geq \sigma^\tau) \leq \frac{n}{\tau}\cdot \frac{1}{n^{h+1}} \leq  n^{-h}
	\end{equation}
	
	Note that $E[\mathcal X] = \sum_{k > 0} k\cdot \mathcal P(\mathcal X=k)$ and $\mathcal P(\mathcal X> 0) = \sum_{k > 0} \mathcal P(\mathcal X=k)$, so $\mathcal P(\mathcal X> 0) \leq E[\mathcal X]$ holds. From Equation \ref{eq:exp X} we obtain $\mathcal P(\mathcal X = 0) \geq 1-n^{-h}$. This is the probability that $P_i < \sigma^\tau$ for all $i=1,\dots, n/\tau$. Applying Lemma \ref{lem:invertible condition}, we obtain our claim.
\end{proof}

Crucially, note that Lemma \ref{lem:invertible} holds for \emph{any} value of $q$ chosen in advance (that is, before choosing $\seed$). 
In the following section we show that, by choosing a uniform prime $q$, with high probability the resulting eq-transform does not generate collisions between the fingerprints of substrings of $S$ and can thus be used to answer substring equality queries correctly.

\subsection{Monte Carlo Eq-Transforms}

In this subsection we formally define Monte Carlo eq-transforms and show how to efficiently compute them. The main idea is to answer substring equalities by comparing fingerprints. One issue with using function $\rks$ for this purpose, however, is that derandomizing it (that is, ensuring that it does not generate collisions among substrings of $S$) is computationally expensive. A common workaround is to introduce an auxiliary hash function that calls $\rks$ only on substrings of $S$ whose lengths are powers of two:

\begin{definition}\label{def: rkspow}
	Let $X\in \Sigma^m$, and let $\ell = 2^{\lfloor \log m\rfloor}$ be the largest power of two smaller than or equal to $m$. The hash function $\rkspow$ is defined as
	$$
	\rkspow(X) = \langle\  \rks(X[0,\ell-1]), \rks(X[m-\ell,m-1])\  \rangle
	$$	
\end{definition}

\begin{definition}
	We say that a hash function $\varphi$ is \emph{collision-free} on $S$ if and only if $S[i,i+\ell] \neq S[i',i'+\ell]$ implies $\varphi(S[i,i+\ell]) \neq \varphi(S[i',i'+\ell])$ for any $0\leq i,i' < n-\ell$.
\end{definition}

At this point, we can formally define our class of Monte Carlo eq-transforms:

\begin{definition}
	We say that $\lambda_{q, \seed, \tau}(S)$ is \emph{Monte Carlo} if and only if $\lambda_{q, \seed, \tau}(S)$ is invertible and $\rkspow$ is collision-free on $S$ with high probability.
\end{definition}

Next, we compute the probability that $\rkspow$ is collision-free given that $q$ is a uniform prime in $\mathcal Z_{\tau,h}$.
We start by proving the following number-theoretic lemma.

\begin{lemma}\label{lemma:primes_in_Z}
	Let $z_{\tau,h}$ denote the number of primes in $\mathcal Z_{\tau,h}$, for $h\geq 2$. If $\tau \geq \frac{12}{5}(h-1)\log_\sigma n$ then $z_{\tau,h} \in \Omega\left( \sigma^{\tau - (h+1)\log_\sigma n} \right)$.
\end{lemma}
\begin{proof}
	The size of $\mathcal Z_{\tau,h}$ is
	
	\begin{equation}\label{eq: Z lower bound}
	|\mathcal Z_{\tau,h}| = \sigma^{\tau}\left( \frac{n^{h-1}}{n^{h-1}-1} \right) - \sigma^{\tau} +1 \geq \sigma^{\tau-(h-1)\log_\sigma n}
	\end{equation}
	
	Let $\pi(x)$ denote the number of primes smaller than $x$. Let moreover $A = \sigma^{\tau}$ and $H = \sigma^{\tau-(h-1)\log_\sigma n}$ be the smallest element contained in $\mathcal Z_{\tau,h}$ and the lower bound for $|\mathcal Z_{\tau,h}|$ stated in Equation \ref{eq: Z lower bound}, respectively. Our aim is to compute an asymptotic lower bound for the number $z_{\tau,h} \geq \pi(A+H) - \pi(A)$ of primes contained in $\mathcal Z_{\tau,h}$. 
	The Prime Number Theorem can be applied to solve this task only if $H \in \Omega(A)$, so we cannot use it in our case. Luckily for us, Heath-Brown~\cite{heath1978differences} proved (see also~\cite{maier1985primes}) that, if $H$ grows at least as quickly as $A^{7/12}$, then: 
	$$
	\lim_{A\rightarrow \infty} \frac{\pi(A+H)-\pi(A)}{H/\ln A} = 1
	$$ 
	Solving $H \geq A^{7/12}$ we get the constraint
	$$
	\tau \geq \frac{12}{5}(h-1)\log_\sigma n 
	$$
	Since $n\rightarrow \infty$ implies that $A\rightarrow \infty$, 
	if the above inequality is satisfied then Heath-Brown's theorem yields the following asymptotic relation as a function of $n$:
	$$
	\pi(A+H)-\pi(A) = \Theta(H/\ln A)
	$$
	The function $H/\ln A$ can be bounded as follows: 
	$$
	H/\ln A = \frac{\sigma^{\tau - (h-1)\log_\sigma n}}{\tau\ln\sigma} \geq \sigma^{\tau - (h+1)\log_\sigma n}
	$$
	where we used the facts that $\ln\sigma < \log\sigma \leq  \log n^{\bigO(1)} \leq n$ and $\tau \leq n$ (which hold for $n$ larger than some constant).
	We finally obtain:
	$$
	z_{\tau,h} \in \Omega\left( \sigma^{\tau - (h+1)\log_\sigma n} \right)
	$$	
	Where we used the $\Omega$ notation (instead of $\Theta$) since Equation \ref{eq: Z lower bound} only implies a lower bound for $z_{\tau,h}$. This will be enough for our purposes.
\end{proof}

Using Lemma \ref{lemma:primes_in_Z}, we can bound the probability of obtaining a collision-free function $\rkspow$:

\begin{lemma}\label{lem:collision-free}
	Given an input string $S\in\Sigma^n$, let:
	\begin{itemize}
		\item $h'\geq 8/3$ be an arbitrarily large constant,
		\item $\tau \geq (3h'-4)\cdot \log_\sigma n$, and
		\item $q$ be a uniform prime number chosen from $\mathcal Z_{\tau,h'}$.
	\end{itemize}
	Then, for any $\seed$ the function $\rkspow$ is collision-free on $S$ with probability at least $(1-n^{-h'})$.
\end{lemma}
\begin{proof}
	First, note that if $\rks$ is collision-free on $S$ then $\rkspow$ is collision-free on $S$ (this follows very easily from the definition of $\rkspow$). 
	
	Let $X \neq Y$, with $|X|=|Y|$, be two substrings of $S$. 
	By definition, $\rks(X) = \rk(X)+\seed\modd q$. As a consequence, $\rk(X) - \rk(Y) \equiv_q \rks(X) - \rks(Y)$. This proves that if $\rk$ is collision-free on $S$ then $\rks$ is collision-free on $S$ and thus, as observed above, $\rkspow$ is collision-free on $S$. 
	Crucially, this holds for any seed $\seed$.
	We therefore concentrate on computing the probability that $\rk$ is collision-free on $S$.
	
	Let $\mathcal C\geq 0$ be the random variable denoting the number of collisions between equal-length substrings of $S$, that is, the number of triples $(i,i',\ell)$ with $i,i'< n -\ell$ such that $S[i,i+\ell] \neq S[i',i'+\ell]$  and $\rk(S[i,i+\ell]) = \rk(S[i',i'+\ell])$.
	Our goal is to compute an upper bound for $\mathcal P(\mathcal C>0)$. 
	Let $X_i^k$ denote the substring of $S$ of length $k$ starting at position $i$. 
	Equivalently, $C$ is the number of substrings $X_i^k \neq X_{i'}^k$ such that $X_i^k \equiv_q X_{i'}^k$.
	
	There is at least one collision ($\mathcal C>0$) iff $X_i^k \equiv_q X_{i'}^k$ for at least one pair $X_i^k \neq X_{i'}^k$, i.e. iff $q$ divides at least one of the numbers $|X_i^k - X_{i'}^k|$ such that $X_i^k \neq X_{i'}^k$. Since $q$ is prime, this happens iff $q$ divides their product
	$z = \prod_{k=1}^{n-1} \prod_{i,i': X_i^k \neq X_{i'}^k} |X_i^k - X_{i'}^k|$.
	Since each $|X_i^k - X_{i'}^k|$ has at most $n$ digits in base $\sigma$ and there are no more than $n^2$ such pairs for every $k$, we have that $z$ has at most $n^4$ digits in base $\sigma$. Written in binary, $z$ has at most $n^4 \lceil\log\sigma\rceil \leq n^5$ digits in base $2$ (again, $\lceil \log\sigma\rceil \leq \log n^{\bigO(1)}+1\leq n$ holds for $n$ larger than some constant). It follows that there cannot be more than $n^5$ distinct primes dividing $z$.

	Note that our choice $\tau \geq (3h'-4)\cdot  \log_\sigma n$ satisfies the inequality  $\tau \geq \frac{12}{5}(h'-1)\log_\sigma n$ of Lemma \ref{lemma:primes_in_Z} for all $h'\geq 8/3$.
	The probability of uniformly picking a prime $q\in\mathcal Z_{\tau,h'}$ dividing $z$ is therefore upper bounded by $n^5/z_{\tau,h'}$, where $z_{\tau,h'} = \Omega\left( \sigma^{\tau - (h'+1)\log_\sigma n} \right)$ is number of primes contained in $\mathcal Z_{\tau,h'}$ computed in Lemma \ref{lemma:primes_in_Z}. A few calculations yield that $n^5/z_{\tau,h'} \in O(n^{-2h'})$, which for $n$ larger than some constant is at most $n^{-h'}$.
\end{proof}

At this point, a large enough block size $\tau$ satisfies both Lemmas \ref{lem:invertible} (invertible transform) and \ref{lem:collision-free} (collision-free hash), allowing us to obtain a Monte Carlo eq-transform:

\begin{lemma}\label{lem:Monte Carlo}
	Given an input string $S\in\Sigma^n$, let:
	\begin{itemize}
		\item $c\geq 8/3$ be an arbitrarily large constant,
		\item $\tau \geq (3c+5)\cdot  \log_\sigma n$,
		\item $q$ be a uniform prime number chosen from $\mathcal Z_{\tau,c+3}$, and
		\item $\seed$ be a uniform integer chosen from $[0,q-1]$.
	\end{itemize}
	Then, $\lambda_{q, \seed, \tau}(S)$ is a Monte Carlo eq-transform. In particular, its associated hash function $\rkspow$ is collision-free with probability at least $1-n^{-c}$.
\end{lemma}
\begin{proof}
	We use constant $h = c+1$ in Lemma \ref{lem:invertible} and $h'=c+3$ in Lemma \ref{lem:collision-free}. 
	Note that $q$ is chosen from $\mathcal Z_{\tau,c+3} = \mathcal Z_{\tau,h+2} = Z_{\tau,h'}$.
	The choice $c \geq 8/3$ implies $h' \geq 8/3$.
	Furthermore, $\tau \geq (3c+5)\cdot  \log_\sigma n = (3h'-4)\cdot \log_\sigma n$. All conditions of Lemma \ref{lem:collision-free} are thus satisfied. By Lemma \ref{lem:collision-free}, $\rkspow$ is collision-free with probability at least $1-n^{-h'}$. By Lemma \ref{lem:invertible}, $\lambda_{q, \seed, \tau}(S)$  is invertible with probability at least $1-n^{-h}$. By union bound the probability that both events hold true simultaneously is at least $1- n^{-h'} +  n^{-h} = 1- n^{-c-3} +  n^{-c-1} \geq 1- 2\cdot n^{-c-1} \geq 1- n^{-c}$ (we can safely assume $n\geq 2$). 
\end{proof}

In order to being able to build our 
data structure \emph{in-place}, one additional major detail needs to be figured out: we need to show that we can extract prime numbers from $\mathcal Z_{\tau,h}$ in place.

\begin{lemma}\label{lemma:in_place_primes}
	Fix two constant $h\geq 2$ and $p>0$, and assume that $\tau \geq \frac{12}{5}(h-1)\log_\sigma n$. Then, in $\bigO(\mathtt{polylog}(n))$ expected time and in $\bigO(1)$ words of space we can find an integer $q\in \mathcal Z_{\tau,h}$ that, with probability at least $1-n^{-p}$, is a prime uniformly distributed among the primes contained in $\mathcal Z_{\tau,h}$.
\end{lemma}
\begin{proof}
	The overall strategy consists in picking a uniform integer from $\mathcal Z_{\tau,h}$ and testing it for primality with a randomized test. If the integer is (probably) prime then we return it, otherwise we repeat the procedure for a predefined maximum amount of trials to guarantee high probability of success.
	If the procedure succeeds, then clearly the result is uniformly distributed among the primes in $\mathcal Z_{\tau,h}$ since every prime in $\mathcal Z_{\tau,h}$ has the same probability to be chosen.

	To test primality we use an in-place version of the Miller-Rabin probabilistic test~\cite{rabin1980probabilistic}. 
	Let $GCD(x, y)$ be the greatest common divisor of $x$ and $y$.
	Given a candidate number $m$ and a positive integer $b < m$, the test verifies the condition $W_m(b)$ defined as
	\begin{enumerate}
		\item $b^{m-1} \not\equiv_m 1$, or
		\item $\exists i$ s.t. $2^i|(m-1)$ and $1<GCD(b^{(m-1)/2^i}-1,m)<m$
	\end{enumerate}   
	If $W_m(b)$ holds true, then $m$ is composite and $b$ is said to be a \emph{witness} of the compositeness of $m$. The key property that makes the test probabilistic is that, if $m$ is composite, then at least $3/4$ of the numbers $1\leq b < m$ are witnesses, i.e. are such that $W_m(b)$ holds true~\cite{rabin1980probabilistic}. It follows that, after $k$ repetitions of the test $W_m(b_1), \dots, W_m(b_k)$ using uniform $b_1, \dots, b_k$, if $m$ is composite then we wrongly identify it as being prime with probability at most $2^{-2k}$. 
	Fix any constant $p>0$.
	We choose $k = (c'\cdot \log_2 n)/2$ for $c' = p+1$. Then, the test is repeated $k\in\bigO(\log n)$ times and the probability of returning a composite number is at most $2^{-2k} = 2^{-c'\cdot \log_2 n} = n^{-c'} = n^{-p-1}$. 
	
	It is not too hard to see that $W_m(b)$ can be tested in $\bigO(\texttt{polylog}(n))$ time while using $\bigO(1)$ words of working space. First, note that the numbers $b< m\in \mathcal Z_{\tau,h}$ for which we test $W_m(b)$ satisfy $b< m\leq n^{\bigO(1)}$, i.e. they can be written using $\bigO(\log n)$ bits. Then, computing $b^y\modd m$, with $y\leq n^{\bigO(1)}$, can be done in-place using the following recurrence given by the fast exponentiation algorithm. Let $|y|$ be the number of bits of $y$ and $y_j$, $1\leq j\leq |y|$, be the $j$-th leftmost bit of $y$. Then, $b^y\modd m = t_{|y|}$, where
	\begin{itemize}
		\item $t_1 = b^{y_1}$
		\item $t_{i+1} = (t_{i}^2\cdot b^{y_i}) \modd m$
	\end{itemize}
	It follows that any $b^y\modd m$, with $y\leq n^{\bigO(1)}$,  can be computed in $\bigO(\log n)$ time using $\bigO(1)$ words of working space. To compute the greatest common divisor in step (2) of the test, note that $GCD(b^{(m-1)/2^i}-1,m) = GCD(b^{(m-1)/2^i}-1\modd m,m)$. This quantity can be computed in $\bigO(\log n)$ time and in $\bigO(1)$ words of working space using the Euclidean algorithm. The GCD has to be computed for at most $\bigO(\log n)$ values of $i$ (i.e. such that $2^i$ divides $m-1$). 
	To sum up, our adaptation of the Miller-Rabin primality test runs in $O(\mathtt{polylog}(n))$ worst-case time and fails in identifying a composite number with probability at most $n^{-p-1}$.
	
	It remains to compute how many candidates $x_1, \dots, x_t \in \mathcal Z_{\tau,h}$ we need to test for primality before giving up and returning a wrong result. In order for our algorithm to run in worst-case time, we need to bound $t$ by some function of $n$. We choose to test $t = c''\log^3 n$ candidates, for  $c'' = p+1$ (we could test less candidates, but for us a polylogarithmic bound will suffice). Lemma \ref{lemma:primes_in_Z} states that the primes density in  $\mathcal Z_{\tau,h}$ is $\Theta((\tau\log\sigma)^{-1}) = \Theta(1/\log n)$. 
	This implies that a group of $\log^2 n$ uniform numbers from $\mathcal Z_{\tau,h}$ contains $\Omega(\log n)$ primes with probability at least $1/2$; the quantity $\Omega(\log n)$ is at least 1 for $n$ larger than some constant. As a result, the probability that  $t = c''\log^3 n$ uniform numbers do not contain any prime is at most $n^{-c''} = n^{-p-1}$. 
	
	Our prime generator works as follows. We test the $t$ candidates for primality, generating them one by one. As soon as the test succeeds (i.e. it finds a potential prime), we return the corresponding candidate and do not proceed with the others. If the test discovers that all $t$ numbers are composite, then we return the last number that we tested (note that this number is composite). 
	Now, let $q\in \mathcal Z_{\tau,h}$ be the number returned by our prime generator. $q$ is composite if (1) the list of tested candidates does not contain any prime, or (2) the primality test failed on a composite number. By union bound, the probability of the union of these two events is upper-bounded by $n^{-c''} + n^{-c'} = 2n^{-p-1} \leq n^{-p}$ (we can safely assume $n\geq 2$).
\end{proof}

In the next theorem we present our data structure. Exhibiting the parameters $q$ and $\seed$ of the hash function $\rkspow$ used internally by the structure will allow us to derandomize it in the next section.

\begin{theorem}\label{th: Monte Carlo structure}
	Let $S\in \Sigma^n$, with $|\Sigma| \leq n^\alpha$ for some constant $\alpha$, be a string representation supporting the extraction and replacement of blocks of $\Theta(\log_\sigma n)$ contiguous characters of $S$ in  time $\packedtime$. Let $d\geq 5/3$ be an arbitrarily large constant. 
	Then, in $O(\packedtime\cdot n/\log_\sigma n)$ worst-case time and using $O(1)$ words of working space on top of $S$, the string $S$ can be replaced with a data structure supporting the following queries: 
	\begin{itemize}
		\item extraction of any block of $O(\log_\sigma n)$ contiguous packed characters of $S$ in $O(\packedtime + \log\log_\sigma n)$ worst-case time and $O(\packedtime)$ memory probes, and 
		\item substring equality queries $S[i,i+\ell] \stackrel{?}{=} S[i',i'+\ell]$ in $O(\packedtime + \log \ell)$ worst-case time and $O(\packedtime)$ memory probes. 
	\end{itemize}
	Using $O(\log n)$ additional precomputed word constants (independent from $S$), the cost of all queries can be improved to $O(\packedtime)$ worst-case time and $O(\packedtime)$ memory probes.
	
	Random access queries are always correct. Substring equality queries are correct with probability at least $1-n^{-d}$.
	Correctness of the data structure is equivalent to the fact that the hash function $\rkspow$, whose parameters $q$ and $\seed$ are accessible from the structure, is collision-free on $S$. This hash function returns fingerprints composed of $3d+2\alpha+10 \in O(1)$ words of $\lceil \log n\rceil$ bits each.
	Finally, the structure can be inverted in-place, restoring $S$, in $O(\packedtime\cdot n/\log_\sigma n)$ worst-case time. 
\end{theorem}
\begin{proof}
	
	Choose any constant $d \geq 5/3$ and fix the block size to $\tau = \lceil (3d+8) \cdot \log_\sigma n \rceil$.
	
	We apply Lemmas \ref{lem:Monte Carlo} and \ref{lemma:in_place_primes} with constants $c,p=d+1 \geq 8/3$ and $h = c+3 > 8/3$. 
	The prime $q$ will be chosen from $\mathcal Z_{\tau,d+4} = \mathcal Z_{\tau,c+3} = \mathcal Z_{\tau,h}$, as required by the two Lemmas.
	The block size $\tau = \lceil (3d+8)\cdot \log_\sigma n \rceil \geq (3c+5)\cdot \log_\sigma n$ satisfies the condition of Lemma \ref{lem:Monte Carlo}. 
	Moreover, $h > 8/3$ implies $h>2$ and $\tau \geq (3c+5)\cdot \log_\sigma n = (3h-4)\cdot \log_\sigma n  \geq \frac{12}{5}(h-1)\log_\sigma n$, thus also the conditions of Lemma \ref{lemma:in_place_primes} are satisfied.

	By Lemma \ref{lemma:in_place_primes}, in $O(\mathtt{polylog}(n)) \subseteq O(n/\log n)$ time we extract a number $q \in \mathcal Z_{\tau,h} = \mathcal Z_{\tau,c+3}$ that with probability at least $1-n^{-p} = 1-n^{-d-1}$ is a uniform prime. Then, we extract a uniform seed $\seed \in [0,q-1]$. If $q$ is prime, by Lemma \ref{lem:Monte Carlo} $\lambda_{q, \seed, \tau}(S)$ is a Monte Carlo eq-transform with probability at least $1-n^{-c} = 1-n^{-d-1}$. 
	Let $q \sim_{U} \mathcal Z_{\tau,h}$ denote the event that $q$ is a uniform prime from $\mathcal Z_{\tau,h}$.
	The probability of obtaining a correct transform is therefore (assume $n\geq 2$):
	$$
	\mathcal P(q \sim_{U} \mathcal Z_{\tau,h}) \cdot \mathcal P(\lambda_{q, \seed, \tau}(S)\ is\ Monte\ Carlo\ |\ q \sim_{U} \mathcal Z_{\tau,h}) \geq (1-n^{-d-1})^2 \geq 1 -2n^{-d-1} \geq  1 -n^{-d}.
	$$
	
	Next, we compute the number $W$ of words of $\lceil \log n\rceil$ bits needed to represent a fingerprint of $\rks$. 
	This number satisfies $W \leq \left\lceil (1+\log F)/\lceil \log n\rceil \right\rceil \leq \log F / \log n + 2$.
	By definition of our prime $q$, each fingerprint can be expressed using at most $\tau+1$ digits in base $\sigma$, i.e. the largest fingerprint $F$ satisfies $F<\sigma^{\tau+1}$. From the definition of the block size $\tau$, we have $\tau+1 \leq (3d+8)\log_\sigma n + 2$. 
	It follows that $\log F < \log \sigma ^{\tau+1} \leq \log \sigma^{(3d+8)\log_\sigma n + 2} = (3d+8)\log n + 2\log \sigma \leq (3d+8+2\alpha)\log n$. To conclude, $W \leq \log F/\log n + 2 \leq 3d + 2\alpha +10$.

	With one scan of $S$ we compute $P_j = \rks(S[0,j\cdot \tau-1])$ for $j = 1,\dots, n/\tau$. At each step $j$ we keep only $P_j$ in memory and discard the previous values. 
	It is easy to see that $P_j$ can be computed from $P_{j-1}$ (with $P_0=\seed$) and $S[(j-1)\cdot \tau, j\cdot \tau-1]$ in constant time, therefore this step can be performed in-place and $O(\packedtime\cdot n/\log_\sigma n)$ worst-case time. 
	If $P_j \geq \sigma^\tau$ for any $1 \leq j \leq n/\tau$, then we do not replace $S$ with its transform and we answer \texttt{false} to all substring equality queries. Otherwise, we replace every block of $S$ with the $\tau$ least significant digits of the corresponding $P_j$, therefore replacing $S$ with $\lambda_{q, \seed, \tau}(S)$. We store also $q$, $\seed$, and $\tau$ on the side in $O(1)$ words. It is easy to see that the transform can be inverted, replacing it with $S$, in $O(\packedtime\cdot n/\log_\sigma n)$ worst-case time and constant working space (with the same strategy used to answer random access, read below). 
	
	Random access boils down to computing the $\tau$-digits integer $X_j = S[(j-1)\cdot \tau,j\cdot \tau-1]$, i.e. the $\tau$ packed characters of the $j$-th block, for any $1 \leq j \leq n/\tau$. As observed earlier, this can be achieved by means of the equality $X_j = (P_{j}-\seed) - (P_{j-1}-\seed)\cdot \sigma^\tau \modd q$ (with $P_0=\seed$). Provided that we pre-compute and store $\sigma^\tau$ in $O(1)$ extra words of space, this calculation takes constant time. Running time is therefore dominated by the cost $O(\packedtime)$ of extracting $P_j$ and $P_{j-1}$ from the transform.
	Once the packed characters $X_j$ of the $j$-th block have been computed, any packed substring of those characters can be retrieved as follows. Extracting the suffix of length $d$ of the block is achieved as $X_j \modd \sigma^{d}$. Any packed substring can then be obtained with a modulo operation (to cut a prefix of the block) followed by an integer division (to shift the result):
	
	$$
	X_j[i,i+d-1] = (X_j \modd \sigma^{\tau-i}) / \sigma^{\tau-i-d}
	$$
	
	The two powers of $\sigma$ have exponent at most $\tau$. They can thus be retrieved in constant time using the $\tau \in O(\log n)$ precomputed word constants $\sigma^1, \dots, \sigma^\tau$ or computed in $O(\log\tau) = O(\log\log_\sigma n)$ time with the fast exponentiation algorithm at no additional space usage.
	
	Solving substring equality queries boils down to computing $\rkspow(X')$ for any substring $X'$ of $S$ of length $\ell = |X'|$. In turn, this requires being able to compute $\rk(X)$ for substrings $X$ of $S$ such that $|X| \leq |X'| = \ell$ is a power of two.
	Using random access, we can extract any block of $\tau$ consecutive digits of $S$. 
	From the definition of $\lambda_{q, \seed, \tau}(S)$, we can moreover retrieve the fingerprint $\rks$ of any block-aligned prefix of $S$ by accessing $\bigO(\log_\sigma n)$ digits of $\lambda_{q, \seed, \tau}(S)$. These operations are sufficient to compute $\rk$ for any prefix $P$ of $S$. We break $P$ as $P'p$, where $P'$ is a sequence of full blocks and $p$ is the prefix of a block. Then: 
	$$
	\rk(P) = (\rks(P')-\seed)\cdot \sigma^{|p|} + p \modd q
	$$
	The value $\rks(P')$ is explicitly stored in the transform. 
	Computing the value $\sigma^{|p|}$ with the fast exponentiation algorithm takes no memory accesses and $\bigO(\log|p|) \subseteq \bigO(\log \ell)$ time. 
	Alternatively, we can pre-compute and store the $O(\tau) \subseteq O(\log n)$ word constants $\sigma^{1}, \dots, \sigma^{\tau}$. Then, $\sigma^{|p|}$ can be retrieved in $\bigO(1)$ time and one memory access.
	If we store these $O(\log n)$ precomputed word constants, then 
	value $p$ can be computed in $O(\packedtime)$ time as shown above (random access). Otherwise, since $p$ is the prefix of some block $B$ retrieved in $O(\packedtime)$ time, we can compute it in $O(\log \ell)$ time with an integer division $p = B/\sigma^{\tau-|p|}$, where $\sigma^{\tau-|p|} = \sigma^\tau / \sigma^{|p|}$ and $\sigma^{|p|}$ was computed earlier in $O(\log \ell)$ time.
	At this point, the fingerprint $\rk(X)$ of any substring $X$ of $S$ whose length $|X|$ is a power of two can be computed as usually done with Karp-Rabin fingerprints, i.e. as a difference (modulo $q$) between the fingerprint of the $S$-prefix ending at the end of $X$ and $\sigma^{|X|}$ times the fingerprint of the $S$-prefix starting before $X$:
	$$
	\rk(S[i,j]) = \rk(S[0,j]) - \rk(S[0,i-1])\cdot \sigma^{j-i+1}\modd q
	$$

	One option for computing $\sigma^{j-i+1}$ is by fast exponentiation in  $O(\log \ell)$ time. 
	Otherwise, 
	note that $j-i+1$ is a power of two, thus it can be expressed as $j-i+1 = 2^e$ for some $0\leq e \leq \log n$. 
	Another option is thus to explicitly store the $O(\log n)$ word constants $\sigma^{2^e}$, for $e=0, \dots, \lfloor \log n \rfloor$ and retrieve $\sigma^{2^e}\modd q$ in constant time with one memory access. Again, note that these words are constants that do not depend on $S$. Taking into account the time $O(\packedtime)$ needed for extracting the required fingerprints of prefixes of $S$, we obtain our claim. 
\end{proof}

\subsection{A Las Vegas Construction Algorithm}

In this section we show how to derandomize the Monte Carlo structure of Theorem \ref{th: Monte Carlo structure} via a Las Vegas construction algorithm: the resulting structure is constructed in worst-case time w.h.p. and is always correct.

\begin{theorem}\label{th: Las Vegas structure}
	Let $S\in \Sigma^n$, with $|\Sigma| \leq n^\alpha$ for some constant $\alpha$, be a string representation supporting the extraction and replacement of blocks of $\Theta(\log_\sigma n)$ contiguous characters of $S$ in  time $\packedtime$. Let $d\geq 5/3$ be an arbitrarily large constant. 
	Then, for any constant $0 < \epsilon \leq 1$ and any parameter $K \in \Omega(n^\epsilon)$, in $O\left(\packedtime\cdot(n^2\log n/K + n/\log_\sigma n )\right)$ worst-case time with probability at least $1-n^{-d}$ and using $K$ rewritable integers in the range $[0,n-1]$ as temporary working space on top of $S$, the string $S$ can be replaced with a data structure supporting the following queries: 
	\begin{itemize}
		\item extraction of any block of $O(\log_\sigma n)$ contiguous packed characters of $S$ in $O(\packedtime + \log\log_\sigma n)$ worst-case time and $O(\packedtime)$ memory probes, and 
		\item substring equality queries $S[i,i+\ell] \stackrel{?}{=} S[i',i'+\ell]$ in $O(\packedtime + \log \ell)$ worst-case time and $O(\packedtime)$ memory probes. 
	\end{itemize}
	Using $O(\log n)$ additional precomputed word constants (independent from $S$), the cost of all queries can be improved to $O(\packedtime)$ worst-case time and $O(\packedtime)$ memory probes.
	
	The data structure itself uses $O(1)$ words on top of the space previously occupied by $S$ (and the $O(\log n)$ word constants, if they are used), always answers correctly to the queries, and it can be inverted in-place, restoring $S$, in $O(\packedtime\cdot n/\log_\sigma n)$ worst-case time. 
\end{theorem}
\begin{proof}
	The idea is to build the structure of Theorem \ref{th: Monte Carlo structure}, obtaining the three parameters $\tau$, $q$, and $\seed$, and checking if $\rkspow$ is collision-free on $S$. If this is not the case, then we invert the structure and re-build it using again  Theorem \ref{th: Monte Carlo structure}. Since this happens with low probability, our final running time will hold in the worst case with high probability.  
	
	To check that $\rkspow$ is collision-free on $S$, we need to ensure that $\rks$ is collision-free among substrings of $S$ whose lengths are powers of two (Definition \ref{def: rkspow}), i.e. for all $S[i,j] \neq S[i',j']$ such that $j-i+1 = j'-i'+1 = 2^e$ for some $e\geq 0$, also $\rks(S[i,j]) \neq \rks(S[i',j'])$ holds. The overall idea is solve the problem iteratively for exponents $e=0,1,\dots, \lfloor\log n \rfloor$: given that the property holds for exponent $e$, we verify that it holds for exponent $e+1$. 
	
	We show how to check collisions between substrings of length $2^e$.
	We conceptually divide the string positions $0,1,\dots, n-1$ into $n/K'$ blocks of size $K' = \frac{K}{2(2\lceil 3d + 2\alpha +10 \rceil + 1)}$: the $i$-th block $B_i$, with $0\leq i < n/K'$, is $B_i = \{i\cdot K', i\cdot K'+1, \dots, i\cdot K' + K' -1\}$ (assume for simplicity that $K'$ divides $n$: if this is not the case, then the last block has less than $K'$ elements; it will be easy to see that this will not affect our bounds).     
	For every pair of distinct blocks $B_i \neq B_{i'}$, we do the following. 
	We build a temporary  sequence $V_{i,i'}$ of length $2K'$ containing all tuples $(\rks(S[j,j+2^e-1]),j)$ such that $j \in B_i \cup B_{i'}$. 
	Each fingerprint $\rks(S[j,j+2^e-1])$ takes exactly $\lceil 3d + 2\alpha +10 \rceil$ words of $\lceil \log n\rceil$ bits each to be stored (see Theorem \ref{th: Monte Carlo structure}). Each of this words can be stored in two integers from the range $[0,n-1]$ (one integer and one bit would actually be enough, but for simplicity we use this upper bound). 
	We thus represent a fingerprint as a tuple of $2\lceil 3d + 2\alpha +10 \rceil$ integers in the range $[0,n-1]$.
	Each index $j$ is moreover an integer in the range $[0,n-1]$. 
	It follows that the sequence $V_{i,i'}$ is composed of tuples totaling $2K' \cdot (2\lceil 3d + 2\alpha +10 \rceil + 1) = K$ integers in the range $[0,n-1]$; this is our temporary working space. 
	Our data structure allows computing $V_{i,i'}$ in $O(K\cdot \packedtime)$ time. 
	At this point, we use the result of Franceschini et al.~\cite{radixsort} to radix-sort in-place the tuples contained in $V_{i,i'}$ (that is, tuples are sorted lexicographically). Since we assume $K \in \Omega(n^\epsilon)$ for some constant $\epsilon > 0$ and since every tuple contained in $V_{i,i'}$ is composed of $O(1)$ integers in the range $[0,n-1]$, in-place radix sort runs in $O(K)$ time and $O(1)$ words of working space. Now, for every pair of adjacent tuples $V_{i,i'}[t] = (j,\rks(S[j,j+2^e-1]))$ and  $V_{i,i'}[t+1] = (j',\rks(S[j',j'+2^e-1]))$ such that $\rks(S[j,j+2^e-1]) = \rks(S[j',j'+2^e-1])$ holds, we check if this is a collision, i.e. if $S[j,j+2^e-1] \neq S[j',j'+2^e-1]$. 
	Note that, without loss of generality, we can assume that $O(\packedtime)$-time queries are available on our structure: it is sufficient to adjust by an arbitrarily small multiplicative constant the parameter $K'$ in order to free $O(\log n)$ words of space (to fit the $O(\log n)$ word constants of Theorem \ref{th: Monte Carlo structure}) while overall still using at most $K$ temporary integers of working space.
	If $e=0$, then the two strings to be compared have length $1$ and the check can be performed in $O(\packedtime)$ time with two random access queries via our data structure. Otherwise ($e>0$), we exploit the fact that we have already checked collisions between all substrings of $S$ of length $2^{e-1}$: $S[j,j+2^e-1] = S[j',j'+2^e-1]$ if and only if their respective two halves match under the hash function $\rks$, i.e. $\rks(S[j,j+2^{e-1}-1]) = \rks(S[j',j'+2^{e-1}-1])$ and $\rks(S[j+2^{e-1},j+2^e-1]) = \rks(S[j'+2^{e-1},j'+2^e-1])$. 
	Using our data structure, these comparisons take  $O(\packedtime)$ time. On the sequence $V_{i,i'}$ we perform at most $2K'-1 \in O(K)$ such comparisons.
	Performing these checks on all pairs of blocks $B_i \neq B_{i'}$ is equivalent to checking all potential collisions among all substrings of $S$ of length $2^e$.
	Summing up, for each $e=0, \dots, \lfloor \log n\rfloor$ we check $O((n/K)^2)$ block pairs, and each verification takes $O(\packedtime\cdot K)$ time. The total running time of the complete verification procedure is thus $O(\packedtime\cdot n^2\log n / K)$, and the total working space is at most $K$ integers in the range $[0,n-1]$. Considering the time needed for building and inverting our structure (Theorem \ref{th: Monte Carlo structure}), we obtain our claim. 
\end{proof}

\subsection{An Optimal String Encoding}

As a first consequence of Theorems \ref{th: Monte Carlo structure} and \ref{th: Las Vegas structure}, in Theorem \ref{th:encoding} we obtain a space-optimal encoding (modulo $O(\log n)$ word constants not depending on $S$)  supporting constant-time random access and substring equality queries.

\begin{theorem}\label{th:encoding}
	There is an encoding representing any string $S$ of length $n$ over alphabet of size $\sigma\leq n^{\bigO(1)}$ in $n\log\sigma + \Theta(\log n)$ bits of space and supporting packed random access and substring equality queries in $\bigO(1)$ worst-case time. The encoding requires $O(\log n)$ precomputed word constants. 
\end{theorem}
\begin{proof}
	Let $S'=\lambda_{q, \seed, \tau}(S)$ be a correct eq-transform for $S$ computed using Theorem \ref{th: Las Vegas structure}, in the fast version requiring $O(\log n)$ additional word constants.
	Our encoding is $\mathcal E_{CP}(S) = \langle q, \seed, \tau, n, \sigma, S' \rangle$. We encode the integers $q, \seed, \tau, n$ using any optimal prefix-free encoding (e.g. Elias delta). We encode $S'$ as follows. Let $S''$ be the string of length $n/\tau$ over the alphabet $[0,\sigma^{\tau}-1]$ formed by representing the blocks of length $\tau$ of $S'$ as numbers in $[0,\sigma^{\tau}-1]$. We store $S''$ with the string representation of Dodis et al.~\cite{dodis2010changing}. This representation takes $\lceil(n/\tau)\log(\sigma^{\tau})\rceil = \lceil n\log\sigma\rceil$ bits of space and supports the extraction of any block of $\Theta(\log_\sigma n)$ packed characters of the original string $S'$ in $\packedtime = O(1)$ time. Also this string representation, similarly to ours, requires the storage of $O(\log n)$ word constants. 
	This is sufficient to achieve constant-time complexity for all operations as described in Theorem \ref{th: Las Vegas structure}. Overall, the size of $\mathcal E_{CP}(S)$ is of $n\log\sigma + \Theta(\log n)$ bits plus $O(\log n)$ word constants. 
\end{proof}

From the theoretical perspective, it is interesting to note that the above encoding could be re-designed to not use the seed $\seed$ at all: we can replace $\rks$ by $\rk$ and store the full fingerprints of every sampled string prefix. The reason is that, as we show next, the representation of Dodis et al.~\cite{dodis2010changing} accommodates for an alphabet increase from $\sigma^\tau$ (the alphabet we use in Theorem \ref{th:encoding} for each block) to $q$ with just $1+O(1/\tau)$ bits of additional space usage \emph{in total}. In particular, an alphabet of size $q$ allows storing inside each block of $S'$ the full fingerprint of the corresponding text prefix, without discarding any digit as done in the original Definition \ref{def:LCE transform}. To see why this is true, note that in this case the representation of Dodis et al.~\cite{dodis2010changing} uses $\lceil (n/\tau)\log q \rceil \leq \lceil (n/\tau)\log \sigma^t \rceil + (n/\tau)\log (q/\sigma^t) + 1$ bits, that is, at most $1+(n/\tau)\log (q/\sigma^t)$ bits on top of $\lceil n\log\sigma\rceil$. Since $q\leq \sigma^\tau \cdot n/(n-1)$, we have $1 + (n/\tau)\log(q/\sigma^\tau) \leq 1+ (1/\tau)\cdot n\log(n/(n-1)) = 1 + O(1/\tau)$. 

\section{Applications to Sparse Text Indexing}\label{sec:in-place results}

In this section we use our new data structures to solve various string processing problems in small space: LCE queries, sparse suffix sorting, sparse LCP array construction, and suffix selection.

\subsection{In-place LCE Queries}

A Longest Common Extension (LCE) query takes as input two suffixes $i$ and $j$ and returns the length of their longest common prefix. 
Note that, given the ability of solving substring equality queries, it is easy to solve LCE queries with exponential search followed by binary search. This strategy requires answering to $\bigO(\log \ell)$ substring equality queries between substrings of length at most $\ell$, where $\ell$ is the final result of the LCE query.

Remember the two space-time trade-offs offered by the structures of Theorems \ref{th: Monte Carlo structure} and \ref{th: Las Vegas structure}.
With \emph{slow LCE queries} on our data structure we denote queries running in $\bigO((\log \ell+\packedtime)\log\ell)$ time and in $\bigO(1)$ words of space on top of the input string. With \emph{fast LCE queries} we denote those running in $\bigO(\packedtime\cdot \log \ell)$ time and requiring $O(\log n)$ word constants of $O(\log n)$ bits on top of the input string. 

In some of our results below, we obtain the space to support fast LCE queries, as well as to derandomize our algorithms, by compressing integer sequences. For this, we adapt a lemma of Franceschini et al.~\cite{radixsort} to work with general integer sequence representations (for example, our lemma works with the representation of Dodis et al.~\cite{dodis2010changing}):


\begin{lemma}\label{lem:compress}
	Let $T$ be a sorted sequence of $k\leq n$ distinct integers from the range $[0,n-1]$ supporting constant-time retrieval and modification of any of its elements. 
	In $O(k)$ time and using $O(1)$ words of working space on top of $T$, we can compress the elements of $T$ so that $\Theta(k\log k/\log n)$ contiguous cells of $T$ are made available. The compressed sequence can be decompressed, restoring the original $T$, in $O(k)$ time and $O(1)$ words of working space.
\end{lemma}
\begin{proof}
	We adapt the approach of Franceschini et al.~\cite{radixsort} in order to remove the assumption that the $k$ integers are represented using $\lceil \log n\rceil$ bits each. 
	We scan the sequence and identify the largest index $i$ such that $T[i] < \lceil (n-1)/2 \rceil$. We store $i$ in one memory word. If $i > k/2$, then we compress the $\Theta(k)$ integers $T[0,i]$ as described below. Otherwise, we subtract $\lceil (n-1)/2 \rceil$ from all integers $T[i+1,k]$ and compress this sequence.
	Let $T[a,b]$ be the sequence to be compressed.
	We divide the sequence in contiguous non-overlapping blocks of length $B = \lceil \log n\rceil+1$ (without loss of generality, assume that $B$ divides the length $b-a+1$ of the sequence).
	Let $T[j,j+B-1]$ be such a block.
	The idea is to "distribute" the last integer $T[j+B-1]$ among the other integers of the block.	
	Let $x_0, \dots, x_{\lceil \log n \rceil-1}$ be the bit-representation of  $T[j+B-1]$.
	The $t$-th most significant bit $x_t$ of $T[j+B-1]$ is stored by overwriting $T[j+t] \leftarrow T[j+t]\cdot 2 + x_t$. 
	Note that $T[j+t]<\lceil (n-1)/2 \rceil$, therefore the result is strictly smaller than $n$ and the operation does not overflow. We compact the resulting sequence (in place and $O(k)$ time) by shifting to its beginning each cell $T[j+B-1]$ that has been freed. 
	Denote as $T[a,a+q-1]$ the $q\in\Theta(k/\log n)$ free cells and by $T[a+q,b]$ the compacted sequence. 
	The free space $T[a,a+q-1]$ can now be used to emulate a bitvector $X[0,2m]$ of length (number of bits) $2m = 2c\cdot k$, for some constant $c>0$.
	Note also that, since we assumed all input integers to be distinct, by the way we constructed it we  have that $T[a+q,b]$ is strictly increasing. Without loss of generality, assume that $m$ divides $n$. 
	Let $t(x) = x\ \mathtt{div}\ (n/m)$: integer $x<n$ belongs to the range $[t(x)\cdot n/m, (t(x)+1)\cdot n/m-1]$.
	Since $T[a+q,b]$ is nondecreasing, we have that $T' = t(T[a+q]), \dots, t(T[b])$ is a nondecreasing sequence in the range $[0,m-1]$. We can thus encode this sequence in the bitvector $X$ by storing in unary the differences between consecutive elements of $T'$. As a result, we can replace $T[j] \leftarrow T[j]\ \mathtt{mod}\ (n/m)$ for all $a+q\leq j \leq b$. The new sequence $T[a+q,b]$ and the bitvector $X$ are enough to restore the original content of $T[a+q,b]$. Since each such $T[j]$ requires $\Theta(\log (n/m))$ bits to be stored, $\Theta(\log m)$ bits of information become available for each $a+q\leq j \leq b$. It follows that, similarly to what we did before, one out of $\Theta(\log n/\log m) = \Theta(\log n/\log k)$ elements of the sequence $T[a+q,b]$ can be "distributed" among its neighbors, thus freeing $\Theta(k \log n/\log k)$ cells (we skip the details for simplicity). These free cells can easily be made contiguous in-place and $O(k)$ time. The whole procedure, as well as its inverse (for decompressing the sequence) runs in-place and $O(k)$ time.
\end{proof}

Even if in Lemma \ref{lem:compress} we do not specify the constant $\delta$, such constant can  be explicitly  calculated after the free space has been made available by measuring it and dividing it by the quantity $k\log k$ (this will be needed later).

\subsection{Monte Carlo In-place SSA and SLCP Construction}

By combining our Monte Carlo data structure with in-place comparison sorting, in-place comparison-based merging~\cite{SS1987}, and compression of integer sequences, in Theorem \ref{th:suffix-sort} we  obtain the first sublinear-time  in-place sparse suffix sorting algorithm.

\begin{theorem}\label{th:suffix-sort}
	Let $S\in \Sigma^n$, with $|\Sigma| \leq n^\alpha$ for some constant $\alpha$, be a string representation supporting the extraction and replacement of blocks of $\Theta(\log_\sigma n)$ contiguous characters of $S$ in  time $\packedtime$.
	Any sequence $B$ of $b$ suffixes of $S$ supporting constant-time access and update queries
	can be suffix-sorted in  $\bigO((n/\log_\sigma n+b\log^2 n)\cdot \packedtime)$ worst-case time using $\bigO(1)$ words of working space on top of $S$ and $B$. The result is correct with probability at least $1-n^{-c}$ for an arbitrarily large constant $c\geq 5/3$ fixed beforehand. The string $S$ is restored before the computation terminates.
\end{theorem}
\begin{proof}
	We first build  our Monte Carlo structure in $\bigO(\packedtime\cdot n/\log_\sigma n)$ time using Theorem \ref{th: Monte Carlo structure}.
	If $b\leq n/\log^4 n$, then we suffix-sort the $b$ string positions plugging slow LCE queries in any in-place comparison sorting algorithm terminating within $\bigO(b\log b)$ comparisons. This requires $\bigO(\packedtime\cdot b\log b\log^2 n)\in\bigO(\packedtime\cdot n/\log n)$ time and $O(1)$ working space. 
	If $b> n/\log^4 n$, then we divide  $B$  into two sub-sequences $B' = B[1,\dots,n/\log^4 n]$ and $B'' = B[n/\log^4 n+1,\dots, b]$ and: 
	\begin{itemize}
		\item Using Lemma \ref{lem:compress} and any fast in-place comparison-based sorting algorithm, we sort and compress $B'$ in-place and in  $O(|B'|\log|B'|) = \bigO((n/\log^4n)\cdot \log(n/\log^4n)) = \bigO(n/\log^3n)$ time. This frees $\omega(\log n)$ integers in the range $[0,n-1]$ of space. We store the $O(\log n)$ word constants required to support fast LCE queries in this space.
		\item We suffix-sort in-place (comparison-based sorting) $B''$ using fast LCE queries. This step terminates in $\bigO(\packedtime\cdot b\log b\log n) \subseteq \bigO(\packedtime\cdot b\cdot \log^2 n)$ time. 
		\item We decompress in-place $B'$ using again Lemma \ref{lem:compress}. Now our structure supports only slow LCE queries.
		\item We suffix sort in-place (comparison-based sorting) $B'$ using slow LCE queries. This step terminates in $\bigO(\packedtime \cdot (n/\log^4 n)\cdot \log(n/\log^4 n)\cdot \log^2 n)\in\bigO(\packedtime\cdot n/\log n)$ time.
		\item We merge $B'$ and $B''$ using slow LCE queries and in-place comparison-based merging~\cite{SS1987}. This step terminates in $\bigO(\packedtime \cdot b\log^2 n)$ time.
		\item Using Theorem \ref{th: Monte Carlo structure}, we invert our data structure and restore $B$ in $\bigO(\packedtime\cdot n/\log_\sigma n)$ time.
	\end{itemize}
	Overall, we use $O(1)$ working space on top of $S$ and $B$. The total running time is bounded by $\bigO((n/\log_\sigma n+b\log^2 n)\cdot \packedtime)$. By Theorem \ref{th: Monte Carlo structure}, the transform (and hence the result) is correct with high probability.
\end{proof}

Note that the generality of our theorem allows using the representation of Dodis et al.~\cite{dodis2010changing} for both $S$ and $B$. In this case, the total space used by the algorithm is just $n\log\sigma + b\log n + \Theta(\log n)$ bits.

We can add a further step to the above-described procedure: after having built the SSA, we can overwrite it with the SLCP array by replacing adjacent SSA entries with the length of their longest common prefix. This step runs in-place and in $\bigO(\packedtime\cdot b\log^2n)$ time using slow LCE queries. This proves:

\begin{theorem}\label{th:sparse LCP}
	Let $S\in \Sigma^n$, with $|\Sigma| \leq n^{O(1)}$, be a string representation supporting the extraction and replacement of blocks of $\Theta(\log_\sigma n)$ contiguous characters of $S$ in  time $\packedtime$.
	Given a sequence $B$ of $b$ suffixes of $S$ supporting constant-time access and update queries, we can replace $B$ with the sparse LCP array relative to $B$ in  $\bigO((n/\log_\sigma n+b\log^2 n)\cdot \packedtime)$ worst-case time using $\bigO(1)$ words of space on top of $S$ and $B$. The result is correct with probability at least $1-n^{-c}$ for an arbitrarily large constant $c\geq 5/3$ fixed beforehand. The string $S$ is restored before the computation terminates.
\end{theorem}

The sparse suffix tree can easily be computed in optimal $\bigO(b)$ space once the SSA and SLCP arrays are available. We therefore obtain the corollary: 

\begin{corollary}\label{th:sparse ST}
	The sparse suffix tree can be computed in $\bigO((n/\log_\sigma n+b\log^2 n)\cdot \packedtime)$ worst-case time using $\bigO(b)$ words of space on top of $S$. The result is correct with high probability.
\end{corollary}

For $b\in o(n/(\packedtime\cdot \log^2n))$ and $\packedtime\in o(\log_\sigma n)$ (that is, packed operations are available), the above algorithms run in sublinear $o(n)$ time.

\subsection{Las Vegas In-place SSA and SLCP Construction}

In order to derandomize Theorems \ref{th:suffix-sort} and \ref{th:sparse LCP}, thus obtaining Las Vegas algorithms, we combine two orthogonal strategies: if the input set $B$ of suffixes is large, then we use our deterministic structure of Theorem \ref{th: Las Vegas structure}. Otherwise, we use our Monte Carlo structure of Theorem \ref{th: Monte Carlo structure} and check correctness of the result. 

\begin{theorem}\label{th:suffix-sort las vegas}
	Let $S\in \Sigma^n$, with $|\Sigma| \leq n^\alpha$ for some constant $\alpha$, be a string representation supporting the extraction and replacement of blocks of $\Theta(\log_\sigma n)$ contiguous characters of $S$ in  time $\packedtime$.
	Any sequence $B$ of $b$ suffixes of $S$ supporting constant-time access and update queries
	can be suffix-sorted in  $\bigO(\packedtime\cdot \min\{b\cdot n/\log_\sigma n,n^2\log n/b\}) \subseteq O(\packedtime \cdot n^{1.5}\sqrt{\log \sigma})$ worst-case time using $\bigO(1)$ words of working space on top of $S$ and $B$. The running time holds with probability at least $1-n^{-c}$ for an arbitrarily large constant $c\geq 5/3$ fixed beforehand, and the returned result is always correct. The string $S$ is restored before the computation terminates.
\end{theorem}
\begin{proof}
	If $b \geq \sqrt{n/\log \sigma}\cdot \log n$, then in $O(b\log b)$ time we 
	sort $B$ with any fast in-place comparison-based sorting algorithm and then we	use Lemma \ref{lem:compress} on $B$ to free $\delta\cdot b\log b/\log n \geq (\delta/2)b$ cells of $B$ in $O(b)$ further time. 
	We use this space to build our deterministic data structure of Theorem \ref{th: Las Vegas structure} with parameter $K = (\delta/2)b \in \Omega(n^{0.5})$. The structure is built in $O(\packedtime\cdot (n^2\log n/b + n/\log_\sigma n))$ worst-case time with probability at least $1-n^{-c}$ for any constant $c\geq 5/3$ fixed beforehand using $O(1)$ words of working space on top of $S$ and $B$. 
	After the structure has been built, we decompress $B$ using again Lemma \ref{lem:compress}.
	We then apply verbatim the procedure described in the proof of Theorem \ref{th:suffix-sort} using our deterministic data structure. The resulting sparse suffix array is always correct. 
	
	If $b < \sqrt{n/\log \sigma}\cdot \log n$, then we compute the sparse suffix array using Theorem \ref{th:suffix-sort}. At this point, we check the result simply by comparing adjacent suffixes (and verifying that they are in the correct lexicographic order) in $O(\packedtime \cdot b\cdot n/\log_\sigma n)$ time using random access queries on $S$. If necessary, we repeat suffix sorting until obtaining a correct result. Since Theorem \ref{th:suffix-sort} yields a correct result with high probability, this procedure runs in $O(\packedtime \cdot b\cdot n/\log_\sigma n)$ worst-case time with high probability and always yields the correct result. 
	
	A quick calculation shows that when $b = \sqrt{n/\log \sigma}\cdot \log n$ the running times of the two procedures (derandomization + suffix sorting / suffix sorting + verification) match and are equal to $O(\packedtime \cdot n^{1.5}\sqrt{\log \sigma})$. Our claim follows easily.
\end{proof}

As noted before Theorem \ref{th:sparse LCP}, after having built the SSA we can overwrite it with the SLCP array by replacing adjacent SSA entries with the length of their longest common prefix. 
The only difference now is that for $b < \sqrt{n/\log \sigma}\cdot \log n$ we compare directly the suffixes using random access in $\bigO(\packedtime\cdot b\cdot n/\log_\sigma n)$ time. For larger $b$, we use slow LCE queries. We obtain:

\begin{theorem}\label{th:sparse LCP las vegas}
	Let $S\in \Sigma^n$, with $\sigma = |\Sigma| \leq n^\alpha$ for some constant $\alpha$, be a string representation supporting the extraction and replacement of blocks of $\Theta(\log_\sigma n)$ contiguous characters of $S$ in  time $\packedtime$.
	Given a sequence $B$ of $b$ suffixes of $S$ supporting constant-time access and update queries, we can replace $B$ with the sparse LCP array relative to $B$ in  $\bigO(\packedtime\cdot \min\{b\cdot n/\log_\sigma n,n^2\log n/ b\}) \subseteq O(\packedtime \cdot n^{1.5}\sqrt{\log \sigma})$ worst-case time using $\bigO(1)$ words of space on top of $S$ and $B$. The running time holds with probability at least $1-n^{-c}$ for an arbitrarily large constant $c\geq 5/3$ fixed beforehand, and the returned result is always correct. The string $S$ is restored before the computation terminates.
\end{theorem}

When the full LCP array has to be computed (that is, $b=n$), the above theorem yields: 

\begin{theorem}\label{th_LCP}
	Let $S\in \Sigma^n$, with $|\Sigma| \leq n^\alpha$ for some constant $\alpha$, be a string representation supporting the extraction and replacement of blocks of $\Theta(\log_\sigma n)$ contiguous characters of $S$ in  time $\packedtime$.
	The Longest Common Prefix array ($LCP$) of $S$ can be computed in $\bigO(\packedtime\cdot n \log n)$ worst-case time using $\bigO(1)$ words of space on top of $S$ and the $LCP$ array itself (the latter stored using any sequence representation supporting constant-time access and update queries). The running time holds with probability at least $1-n^{-c}$ for an arbitrarily large constant $c\geq 5/3$ fixed beforehand, and the returned result is always correct. The string $S$ is restored before the computation terminates.
\end{theorem}

\subsection{In-place Suffix Selection}

The selection problem on integers admits fast in-place solutions in the read-only model~\cite{MUNRO1996311}. The analogous in-place suffix selection problem is, however, still open.  
In Theorem \ref{th:suf sel} we provide the first in-place sub-quadratic algorithm solving this problem in the restore model: given a rewritable string $S$ (to be restored before the computation terminates) and an index $0\leq i <n$, output the string position corresponding to the $i$-th lexicographically smallest suffix. 
Our idea is to use a variant of the quick-select algorithm (i.e. solving the selection problem on integers) operating with the lexicographical ordering of the suffixes. 
We note that the algorithm in~\cite{crochemore2013constant}---replacing in-place and quadratic time the string with its Burrows-Wheeler transform---can be easily adapted to solve this problem in-place and $\bigO(n^2)$ time: after building the BWT, we apply the LF mapping from position $i$ until reaching the BWT terminator character. The number of LF steps is equal to the length of the $i$-th smallest string suffix. Each LF step runs in $\bigO(n)$ time (i.e. a simple scan), therefore the overall algorithm runs in-place and $\bigO(n^2)$ time. Theorem \ref{th:suf sel} reduces this time to $\bigO(\packedtime\cdot n\log^4n)$ with high probability:

\begin{theorem}\label{th:suf sel}
	Let $S\in \Sigma^n$, with $|\Sigma| \leq n^\alpha$ for some constant $\alpha$, be a string representation supporting the extraction and replacement of blocks of $\Theta(\log_\sigma n)$ contiguous characters of $S$ in  time $\packedtime$.
	Given an integer $i<n$ we can find the $i$-th lexicographically smallest suffix of $S$ in $\bigO(\packedtime\cdot n\log^4n)$ worst-case time using $\bigO(1)$ words of working space on top of $S$. The result is correct with high probability. The string $S$ is restored before the computation terminates.
\end{theorem}
\begin{proof}
	We build  our Monte Carlo structure using  Theorem \ref{th: Monte Carlo structure}. In our procedure below, we use slow LCE queries to compare pairs of suffixes lexicographically.
	We scan $S$ and find the lexicographically smallest and largest suffixes $i_{\min}$, $i_{\max}$ in $\bigO(n\log^2n)$ time: starting with $i_{\min} = i_{\max} = 0$, for every $0\leq j < n$ we compare the $j$-th suffix with the $i_{\min}$-th and $i_{\max}$-th suffixes and determine whether $j$ is the new $i_{\min}$ or $i_{\max}$. Then, we scan again $S$ and count the number $m$ of suffixes inside the lexicographic range $[i_{\min}, i_{\max}]$ (at the beginning, $m=n$). 
	We pick a uniform random number $r$ in $[1,m]$, scan again $S$, and select the $r$-th suffix $i_r$ falling inside the lexicographic range $[i_{\min}, i_{\max}]$ that we see during the scan. 
	Let $R$ be the lexicographic range of suffixes --- either $[i_{\min},i_r)$ or $[i_r,i_{\max}]$ --- the $i$-th smallest suffix belongs to. 
	In order to guarantee $O(\log n)$ recursive calls in the worst case, we need $R$ to be of size at most $m/2$. By repeating the procedure $O(\log n)$ times, we will find such $i_r$ with high probability. If we do not find it, we just stop the algorithm and return an arbitrary result (possibly, wrong). We perform at most $\bigO(\log n)$ recursive calls, therefore our algorithm runs in $\bigO(n\log^4n)$ worst-case time and returns the correct result with high probability.	
\end{proof}

\section{Conclusions}

We have presented a new in-place data structure supporting efficient substring equality queries.
The string can be replaced with our data structure, as well as restored, in-place and in optimal packed time. We also showed how to  derandomize the data structure using small additional working space. 
Our structure represents a powerful tool for solving in-place a wide range of string processing problems: using this technique we provided the first in-place and subquadratic-time (indeed, even sublinear in some cases) algorithms for the sparse suffix sorting, sparse LCP construction, and suffix selection problems. Using particular input representations, our suffix sorting and LCP construction algorithms use just $n\log\sigma + b\log n + \Theta(\log n)$ bits of \emph{total} space (that is, including input, output, and working space). 

Our work leaves several intriguing open problems. First of all, can we remove the need of storing $\Theta(\log n)$ word constants in order to support fast substring equality (or even just LCE) queries? Other major open problems are whether our string processing algorithms can be derandomized in-place with a faster procedure (can we suffix sort an arbitrary set of suffixes in-place and $o(n^{1.5})$ time, always returning the correct result?), and whether we can achieve the same time-space bounds in the read-only model (i.e. where the input string is not allowed to be overwritten). 
Suffix-sorting an arbitrary subset of the string's suffixes in subquadratic time using $\bigO(1)$ space and without overwriting the original string seems to be a hard task; we conjecture that $\omega(1)$ space is needed in this model of computation in order to achieve subquadratic running time.

\subsection*{Acknowledgements}

I would like to thank the anonymous reviewers: their observations greatly improved the presentation of the results, as well as some of the results themselves (in particular, w.h.p. time bounds and working space of the derandomization procedure).

\bibliographystyle{ACM-Reference-Format}
\bibliography{paper}

%
%
 \bibliographystyle{splncs04}
 \bibliography{paper}


\begin{thebibliography}{38}


\ifx \showCODEN    \undefined \def \showCODEN     #1{\unskip}     \fi
\ifx \showDOI      \undefined \def \showDOI       #1{#1}\fi
\ifx \showISBNx    \undefined \def \showISBNx     #1{\unskip}     \fi
\ifx \showISBNxiii \undefined \def \showISBNxiii  #1{\unskip}     \fi
\ifx \showISSN     \undefined \def \showISSN      #1{\unskip}     \fi
\ifx \showLCCN     \undefined \def \showLCCN      #1{\unskip}     \fi
\ifx \shownote     \undefined \def \shownote      #1{#1}          \fi
\ifx \showarticletitle \undefined \def \showarticletitle #1{#1}   \fi
\ifx \showURL      \undefined \def \showURL       {\relax}        \fi
\providecommand\bibfield[2]{#2}
\providecommand\bibinfo[2]{#2}
\providecommand\natexlab[1]{#1}
\providecommand\showeprint[2][]{arXiv:#2}

\bibitem[\protect\citeauthoryear{Baeza-Yates and Gonnet}{Baeza-Yates and
  Gonnet}{1992}]%
        {baeza1992new}
\bibfield{author}{\bibinfo{person}{Ricardo Baeza-Yates} {and}
  \bibinfo{person}{Gaston~H. Gonnet}.} \bibinfo{year}{1992}\natexlab{}.
\newblock \showarticletitle{A New Approach to Text Searching}.
\newblock \bibinfo{journal}{\emph{Commun. ACM}} \bibinfo{volume}{35},
  \bibinfo{number}{10} (\bibinfo{date}{Oct.} \bibinfo{year}{1992}),
  \bibinfo{pages}{74--82}.
\newblock
\showISSN{0001-0782}


\bibitem[\protect\citeauthoryear{Bentley and Yao}{Bentley and Yao}{1976}]%
        {bentley1976almost}
\bibfield{author}{\bibinfo{person}{Jon~Louis Bentley} {and}
  \bibinfo{person}{Andrew Chi-Chih Yao}.} \bibinfo{year}{1976}\natexlab{}.
\newblock \showarticletitle{An almost optimal algorithm for unbounded
  searching}.
\newblock \bibinfo{journal}{\emph{Information processing letters}}
  \bibinfo{volume}{5}, \bibinfo{number}{SLAC-PUB-1679} (\bibinfo{year}{1976}).
\newblock


\bibitem[\protect\citeauthoryear{Bille, Fischer, G{\o}rtz, Kopelowitz, Sach,
  and Vildh{\o}j}{Bille et~al\mbox{.}}{2016}]%
        {bille2013sparse}
\bibfield{author}{\bibinfo{person}{Philip Bille}, \bibinfo{person}{Johannes
  Fischer}, \bibinfo{person}{Inge~Li G{\o}rtz}, \bibinfo{person}{Tsvi
  Kopelowitz}, \bibinfo{person}{Benjamin Sach}, {and}
  \bibinfo{person}{Hjalte~Wedel Vildh{\o}j}.} \bibinfo{year}{2016}\natexlab{}.
\newblock \showarticletitle{Sparse Text Indexing in Small Space}.
\newblock \bibinfo{journal}{\emph{ACM Trans. Algorithms}} \bibinfo{volume}{12},
  \bibinfo{number}{3}, Article \bibinfo{articleno}{39} (\bibinfo{date}{April}
  \bibinfo{year}{2016}), \bibinfo{numpages}{19}~pages.
\newblock
\showISSN{1549-6325}


\bibitem[\protect\citeauthoryear{Bille, G{\o}rtz, Knudsen, Lewenstein, and
  Vildh{\o}j}{Bille et~al\mbox{.}}{2015}]%
        {bille2015longest}
\bibfield{author}{\bibinfo{person}{Philip Bille}, \bibinfo{person}{Inge~Li
  G{\o}rtz}, \bibinfo{person}{Mathias B{\ae}k~Tejs Knudsen},
  \bibinfo{person}{Moshe Lewenstein}, {and} \bibinfo{person}{Hjalte~Wedel
  Vildh{\o}j}.} \bibinfo{year}{2015}\natexlab{}.
\newblock \showarticletitle{Longest Common Extensions in Sublinear Space}. In
  \bibinfo{booktitle}{\emph{Combinatorial Pattern Matching}},
  \bibfield{editor}{\bibinfo{person}{Ferdinando Cicalese}, \bibinfo{person}{Ely
  Porat}, {and} \bibinfo{person}{Ugo Vaccaro}} (Eds.).
  \bibinfo{publisher}{Springer International Publishing},
  \bibinfo{address}{Cham}, \bibinfo{pages}{65--76}.
\newblock


\bibitem[\protect\citeauthoryear{Bille, G{\o}rtz, Sach, and Vildh{\o}j}{Bille
  et~al\mbox{.}}{2012}]%
        {bille2014time}
\bibfield{author}{\bibinfo{person}{Philip Bille}, \bibinfo{person}{Inge~Li
  G{\o}rtz}, \bibinfo{person}{Benjamin Sach}, {and}
  \bibinfo{person}{Hjalte~Wedel Vildh{\o}j}.} \bibinfo{year}{2012}\natexlab{}.
\newblock \showarticletitle{Time-Space Trade-Offs for Longest Common
  Extensions}. In \bibinfo{booktitle}{\emph{Combinatorial Pattern Matching}},
  \bibfield{editor}{\bibinfo{person}{Juha K{\"a}rkk{\"a}inen} {and}
  \bibinfo{person}{Jens Stoye}} (Eds.). \bibinfo{publisher}{Springer Berlin
  Heidelberg}, \bibinfo{address}{Berlin, Heidelberg},
  \bibinfo{pages}{293--305}.
\newblock
\showISBNx{978-3-642-31265-6}


\bibitem[\protect\citeauthoryear{Birenzwige, Golan, and Porat}{Birenzwige
  et~al\mbox{.}}{2020}]%
        {Birenzwige20}
\bibfield{author}{\bibinfo{person}{Or Birenzwige}, \bibinfo{person}{Shay
  Golan}, {and} \bibinfo{person}{Ely Porat}.} \bibinfo{year}{2020}\natexlab{}.
\newblock \showarticletitle{Locally Consistent Parsing for Text Indexing in
  Small Space}. In \bibinfo{booktitle}{\emph{Proceedings of the Thirty-First
  Annual ACM-SIAM Symposium on Discrete Algorithms}}
  \emph{(\bibinfo{series}{SODA ’20})}. \bibinfo{publisher}{Society for
  Industrial and Applied Mathematics}, \bibinfo{address}{USA},
  \bibinfo{pages}{607–626}.
\newblock


\bibitem[\protect\citeauthoryear{Crochemore, Grossi, K{\"a}rkk{\"a}inen, and
  Landau}{Crochemore et~al\mbox{.}}{2013}]%
        {crochemore2013constant}
\bibfield{author}{\bibinfo{person}{Maxime Crochemore}, \bibinfo{person}{Roberto
  Grossi}, \bibinfo{person}{Juha K{\"a}rkk{\"a}inen}, {and}
  \bibinfo{person}{Gad~M. Landau}.} \bibinfo{year}{2013}\natexlab{}.
\newblock \showarticletitle{A Constant-Space Comparison-Based Algorithm for
  Computing the Burrows-Wheeler Transform}. In
  \bibinfo{booktitle}{\emph{Combinatorial Pattern Matching}},
  \bibfield{editor}{\bibinfo{person}{Johannes Fischer} {and}
  \bibinfo{person}{Peter Sanders}} (Eds.). \bibinfo{publisher}{Springer Berlin
  Heidelberg}, \bibinfo{address}{Berlin, Heidelberg}, \bibinfo{pages}{74--82}.
\newblock
\showISBNx{978-3-642-38905-4}


\bibitem[\protect\citeauthoryear{Dodis, Patrascu, and Thorup}{Dodis
  et~al\mbox{.}}{2010}]%
        {dodis2010changing}
\bibfield{author}{\bibinfo{person}{Yevgeniy Dodis}, \bibinfo{person}{Mihai
  Patrascu}, {and} \bibinfo{person}{Mikkel Thorup}.}
  \bibinfo{year}{2010}\natexlab{}.
\newblock \showarticletitle{Changing Base Without Losing Space}. In
  \bibinfo{booktitle}{\emph{Proceedings of the Forty-second ACM Symposium on
  Theory of Computing}} \emph{(\bibinfo{series}{STOC '10})}.
  \bibinfo{publisher}{ACM}, \bibinfo{address}{New York, NY, USA},
  \bibinfo{pages}{593--602}.
\newblock
\showISBNx{978-1-4503-0050-6}


\bibitem[\protect\citeauthoryear{Fischer, I., and K{\"o}ppl}{Fischer
  et~al\mbox{.}}{2016}]%
        {fischer2016deterministic}
\bibfield{author}{\bibinfo{person}{Johannes Fischer}, \bibinfo{person}{Tomohiro
  I.}, {and} \bibinfo{person}{Dominik K{\"o}ppl}.}
  \bibinfo{year}{2016}\natexlab{}.
\newblock \showarticletitle{Deterministic Sparse Suffix Sorting on Rewritable
  Texts}. In \bibinfo{booktitle}{\emph{LATIN 2016: Theoretical Informatics}},
  \bibfield{editor}{\bibinfo{person}{Evangelos Kranakis},
  \bibinfo{person}{Gonzalo Navarro}, {and} \bibinfo{person}{Edgar Ch{\'a}vez}}
  (Eds.). \bibinfo{publisher}{Springer Berlin Heidelberg},
  \bibinfo{address}{Berlin, Heidelberg}, \bibinfo{pages}{483--496}.
\newblock


\bibitem[\protect\citeauthoryear{Franceschini and Muthukrishnan}{Franceschini
  and Muthukrishnan}{2007a}]%
        {franceschini2007place}
\bibfield{author}{\bibinfo{person}{Gianni Franceschini} {and}
  \bibinfo{person}{Shan Muthukrishnan}.} \bibinfo{year}{2007}\natexlab{a}.
\newblock \showarticletitle{In-place Suffix Sorting}. In
  \bibinfo{booktitle}{\emph{Proceedings of the 34th International Conference on
  Automata, Languages and Programming}} \emph{(\bibinfo{series}{ICALP'07})}.
  \bibinfo{publisher}{Springer-Verlag}, \bibinfo{address}{Berlin, Heidelberg},
  \bibinfo{pages}{533--545}.
\newblock
\showISBNx{3-540-73419-8, 978-3-540-73419-2}


\bibitem[\protect\citeauthoryear{Franceschini and Muthukrishnan}{Franceschini
  and Muthukrishnan}{2007b}]%
        {franceschini2007optimal}
\bibfield{author}{\bibinfo{person}{Gianni Franceschini} {and}
  \bibinfo{person}{Shan Muthukrishnan}.} \bibinfo{year}{2007}\natexlab{b}.
\newblock \showarticletitle{Optimal Suffix Selection}. In
  \bibinfo{booktitle}{\emph{Proceedings of the Thirty-ninth Annual ACM
  Symposium on Theory of Computing}} \emph{(\bibinfo{series}{STOC '07})}.
  \bibinfo{publisher}{ACM}, \bibinfo{address}{New York, NY, USA},
  \bibinfo{pages}{328--337}.
\newblock
\showISBNx{978-1-59593-631-8}


\bibitem[\protect\citeauthoryear{Franceschini, Muthukrishnan, and
  P{\v{a}}tra{\c{s}}cu}{Franceschini et~al\mbox{.}}{2007}]%
        {radixsort}
\bibfield{author}{\bibinfo{person}{Gianni Franceschini}, \bibinfo{person}{Shan
  Muthukrishnan}, {and} \bibinfo{person}{Mihai P{\v{a}}tra{\c{s}}cu}.}
  \bibinfo{year}{2007}\natexlab{}.
\newblock \showarticletitle{Radix Sorting with No Extra Space}. In
  \bibinfo{booktitle}{\emph{Algorithms -- ESA 2007}},
  \bibfield{editor}{\bibinfo{person}{Lars Arge}, \bibinfo{person}{Michael
  Hoffmann}, {and} \bibinfo{person}{Emo Welzl}} (Eds.).
  \bibinfo{publisher}{Springer Berlin Heidelberg}, \bibinfo{address}{Berlin,
  Heidelberg}, \bibinfo{pages}{194--205}.
\newblock


\bibitem[\protect\citeauthoryear{G{\'a}l and Miltersen}{G{\'a}l and
  Miltersen}{2003}]%
        {gal2003cell}
\bibfield{author}{\bibinfo{person}{Anna G{\'a}l} {and}
  \bibinfo{person}{Peter~Bro Miltersen}.} \bibinfo{year}{2003}\natexlab{}.
\newblock \showarticletitle{The Cell Probe Complexity of Succinct Data
  Structures}. In \bibinfo{booktitle}{\emph{Automata, Languages and
  Programming}}, \bibfield{editor}{\bibinfo{person}{Jos C.~M. Baeten},
  \bibinfo{person}{Jan~Karel Lenstra}, \bibinfo{person}{Joachim Parrow}, {and}
  \bibinfo{person}{Gerhard~J. Woeginger}} (Eds.). \bibinfo{publisher}{Springer
  Berlin Heidelberg}, \bibinfo{address}{Berlin, Heidelberg},
  \bibinfo{pages}{332--344}.
\newblock
\showISBNx{978-3-540-45061-0}


\bibitem[\protect\citeauthoryear{Gawrychowski and Kociumaka}{Gawrychowski and
  Kociumaka}{2017}]%
        {gawrychowski2017sparse}
\bibfield{author}{\bibinfo{person}{Pawe\l\ Gawrychowski} {and}
  \bibinfo{person}{Tomasz Kociumaka}.} \bibinfo{year}{2017}\natexlab{}.
\newblock \showarticletitle{Sparse Suffix Tree Construction in Optimal Time and
  Space}. In \bibinfo{booktitle}{\emph{Proceedings of the Twenty-Eighth Annual
  ACM-SIAM Symposium on Discrete Algorithms}} \emph{(\bibinfo{series}{SODA
  '17})}. \bibinfo{publisher}{Society for Industrial and Applied Mathematics},
  \bibinfo{address}{Philadelphia, PA, USA}, \bibinfo{pages}{425--439}.
\newblock


\bibitem[\protect\citeauthoryear{Gog and Ohlebusch}{Gog and Ohlebusch}{2011}]%
        {gog2011fast}
\bibfield{author}{\bibinfo{person}{Simon Gog} {and} \bibinfo{person}{Enno
  Ohlebusch}.} \bibinfo{year}{2011}\natexlab{}.
\newblock \showarticletitle{Fast and Lightweight LCP-array Construction
  Algorithms}. In \bibinfo{booktitle}{\emph{Proceedings of the Meeting on
  Algorithm Engineering \& Expermiments}} \emph{(\bibinfo{series}{ALENEX
  '11})}. \bibinfo{publisher}{Society for Industrial and Applied Mathematics},
  \bibinfo{address}{Philadelphia, PA, USA}, \bibinfo{pages}{25--34}.
\newblock


\bibitem[\protect\citeauthoryear{Golynski}{Golynski}{2009}]%
        {golynski2009cell}
\bibfield{author}{\bibinfo{person}{Alexander Golynski}.}
  \bibinfo{year}{2009}\natexlab{}.
\newblock \showarticletitle{Cell Probe Lower Bounds for Succinct Data
  Structures}. In \bibinfo{booktitle}{\emph{Proceedings of the Twentieth Annual
  ACM-SIAM Symposium on Discrete Algorithms}} \emph{(\bibinfo{series}{SODA
  '09})}. \bibinfo{publisher}{Society for Industrial and Applied Mathematics},
  \bibinfo{address}{Philadelphia, PA, USA}, \bibinfo{pages}{625--634}.
\newblock


\bibitem[\protect\citeauthoryear{Gonnet, Baeza-Yates, and Snider}{Gonnet
  et~al\mbox{.}}{1992}]%
        {gonnet1992new}
\bibfield{author}{\bibinfo{person}{Gaston~H. Gonnet},
  \bibinfo{person}{Ricardo~A. Baeza-Yates}, {and} \bibinfo{person}{Tim
  Snider}.} \bibinfo{year}{1992}\natexlab{}.
\newblock \showarticletitle{Information Retrieval}.
\newblock \bibinfo{publisher}{Prentice-Hall, Inc.}, \bibinfo{address}{Upper
  Saddle River, NJ, USA}, Chapter New Indices for Text: PAT Trees and PAT
  Arrays, \bibinfo{pages}{66--82}.
\newblock
\showISBNx{0-13-463837-9}


\bibitem[\protect\citeauthoryear{Hansen}{Hansen}{1994}]%
        {Hansen1994}
\bibfield{author}{\bibinfo{person}{Per~Brinch Hansen}.}
  \bibinfo{year}{1994}\natexlab{}.
\newblock \showarticletitle{Multiple-length division revisited: A tour of the
  minefield}.
\newblock \bibinfo{journal}{\emph{Software: Practice and Experience}}
  \bibinfo{volume}{24}, \bibinfo{number}{6} (\bibinfo{year}{1994}),
  \bibinfo{pages}{579--601}.
\newblock


\bibitem[\protect\citeauthoryear{Heath-Brown}{Heath-Brown}{1978}]%
        {heath1978differences}
\bibfield{author}{\bibinfo{person}{DR Heath-Brown}.}
  \bibinfo{year}{1978}\natexlab{}.
\newblock \showarticletitle{The differences between consecutive primes}.
\newblock \bibinfo{journal}{\emph{Journal of the London Mathematical Society}}
  \bibinfo{volume}{2}, \bibinfo{number}{1} (\bibinfo{year}{1978}),
  \bibinfo{pages}{7--13}.
\newblock


\bibitem[\protect\citeauthoryear{I, K\"arkk\"ainen, and Kempa}{I
  et~al\mbox{.}}{2014}]%
        {karkkainen2014faster}
\bibfield{author}{\bibinfo{person}{Tomohiro I}, \bibinfo{person}{Juha
  K\"arkk\"ainen}, {and} \bibinfo{person}{Dominik Kempa}.}
  \bibinfo{year}{2014}\natexlab{}.
\newblock \showarticletitle{{Faster Sparse Suffix Sorting}}. In
  \bibinfo{booktitle}{\emph{31st International Symposium on Theoretical Aspects
  of Computer Science}}. Schloss Dagstuhl-Leibniz-Zentrum fuer Informatik,
  \bibinfo{pages}{386--396}.
\newblock


\bibitem[\protect\citeauthoryear{K\"{a}rkk\"{a}inen, Sanders, and
  Burkhardt}{K\"{a}rkk\"{a}inen et~al\mbox{.}}{2006}]%
        {karkkainen2006linear}
\bibfield{author}{\bibinfo{person}{Juha K\"{a}rkk\"{a}inen},
  \bibinfo{person}{Peter Sanders}, {and} \bibinfo{person}{Stefan Burkhardt}.}
  \bibinfo{year}{2006}\natexlab{}.
\newblock \showarticletitle{Linear Work Suffix Array Construction}.
\newblock \bibinfo{journal}{\emph{J. ACM}} \bibinfo{volume}{53},
  \bibinfo{number}{6} (\bibinfo{date}{Nov.} \bibinfo{year}{2006}),
  \bibinfo{pages}{918--936}.
\newblock
\showISSN{0004-5411}


\bibitem[\protect\citeauthoryear{K\"{a}rkk\"{a}inen and
  Ukkonen}{K\"{a}rkk\"{a}inen and Ukkonen}{1996}]%
        {karkkainen1996sparse}
\bibfield{author}{\bibinfo{person}{Juha K\"{a}rkk\"{a}inen} {and}
  \bibinfo{person}{Esko Ukkonen}.} \bibinfo{year}{1996}\natexlab{}.
\newblock \showarticletitle{Sparse Suffix Trees}. In
  \bibinfo{booktitle}{\emph{Proceedings of the Second Annual International
  Conference on Computing and Combinatorics}} \emph{(\bibinfo{series}{COCOON
  '96})}. \bibinfo{publisher}{Springer-Verlag}, \bibinfo{address}{London, UK,
  UK}, \bibinfo{pages}{219--230}.
\newblock
\showISBNx{3-540-61332-3}


\bibitem[\protect\citeauthoryear{Karp and Rabin}{Karp and Rabin}{1987}]%
        {karp1987efficient}
\bibfield{author}{\bibinfo{person}{Richard~M. Karp} {and}
  \bibinfo{person}{Michael~O. Rabin}.} \bibinfo{year}{1987}\natexlab{}.
\newblock \showarticletitle{Efficient Randomized Pattern-matching Algorithms}.
\newblock \bibinfo{journal}{\emph{IBM J. Res. Dev.}} \bibinfo{volume}{31},
  \bibinfo{number}{2} (\bibinfo{date}{March} \bibinfo{year}{1987}),
  \bibinfo{pages}{249--260}.
\newblock
\showISSN{0018-8646}


\bibitem[\protect\citeauthoryear{Kempa and Kociumaka}{Kempa and
  Kociumaka}{2019}]%
        {Kempa19}
\bibfield{author}{\bibinfo{person}{Dominik Kempa} {and} \bibinfo{person}{Tomasz
  Kociumaka}.} \bibinfo{year}{2019}\natexlab{}.
\newblock \showarticletitle{String Synchronizing Sets: Sublinear-Time BWT
  Construction and Optimal LCE Data Structure}. In
  \bibinfo{booktitle}{\emph{Proceedings of the 51st Annual ACM SIGACT Symposium
  on Theory of Computing}} \emph{(\bibinfo{series}{STOC 2019})}.
  \bibinfo{publisher}{Association for Computing Machinery},
  \bibinfo{address}{New York, NY, USA}, \bibinfo{pages}{756–767}.
\newblock
\showISBNx{9781450367059}


\bibitem[\protect\citeauthoryear{Knuth}{Knuth}{1997}]%
        {knuthbook}
\bibfield{author}{\bibinfo{person}{Donald~E. Knuth}.}
  \bibinfo{year}{1997}\natexlab{}.
\newblock \bibinfo{booktitle}{\emph{The Art of Computer Programming, Volume 2
  (3rd Ed.): Seminumerical Algorithms}}.
\newblock \bibinfo{publisher}{Addison-Wesley Longman Publishing Co., Inc.},
  \bibinfo{address}{Boston, MA, USA}.
\newblock
\showISBNx{0-201-89684-2}


\bibitem[\protect\citeauthoryear{Kosolobov}{Kosolobov}{2017}]%
        {kosolobov2017tight}
\bibfield{author}{\bibinfo{person}{Dmitry Kosolobov}.}
  \bibinfo{year}{2017}\natexlab{}.
\newblock \showarticletitle{Tight lower bounds for the longest common extension
  problem}.
\newblock \bibinfo{journal}{\emph{Inform. Process. Lett.}}
  \bibinfo{volume}{125} (\bibinfo{year}{2017}), \bibinfo{pages}{26 -- 29}.
\newblock
\showISSN{0020-0190}


\bibitem[\protect\citeauthoryear{Li, Li, and Huo}{Li et~al\mbox{.}}{2018}]%
        {li2018optimal}
\bibfield{author}{\bibinfo{person}{Zhize Li}, \bibinfo{person}{Jian Li}, {and}
  \bibinfo{person}{Hongwei Huo}.} \bibinfo{year}{2018}\natexlab{}.
\newblock \showarticletitle{Optimal In-Place Suffix Sorting}. In
  \bibinfo{booktitle}{\emph{String Processing and Information Retrieval}},
  \bibfield{editor}{\bibinfo{person}{Travis Gagie}, \bibinfo{person}{Alistair
  Moffat}, \bibinfo{person}{Gonzalo Navarro}, {and} \bibinfo{person}{Ernesto
  Cuadros-Vargas}} (Eds.). \bibinfo{publisher}{Springer International
  Publishing}, \bibinfo{address}{Cham}, \bibinfo{pages}{268--284}.
\newblock
\showISBNx{978-3-030-00479-8}


\bibitem[\protect\citeauthoryear{Louza, Gagie, and Telles}{Louza
  et~al\mbox{.}}{2017}]%
        {LOUZA201714}
\bibfield{author}{\bibinfo{person}{Felipe~A. Louza}, \bibinfo{person}{Travis
  Gagie}, {and} \bibinfo{person}{Guilherme~P. Telles}.}
  \bibinfo{year}{2017}\natexlab{}.
\newblock \showarticletitle{Burrows-Wheeler transform and LCP array
  construction in constant space}.
\newblock \bibinfo{journal}{\emph{Journal of Discrete Algorithms}}
  \bibinfo{volume}{42} (\bibinfo{year}{2017}), \bibinfo{pages}{14 -- 22}.
\newblock
\showISSN{1570-8667}


\bibitem[\protect\citeauthoryear{Maier}{Maier}{1985}]%
        {maier1985primes}
\bibfield{author}{\bibinfo{person}{Helmut Maier}.}
  \bibinfo{year}{1985}\natexlab{}.
\newblock \showarticletitle{Primes in short intervals.}
\newblock \bibinfo{journal}{\emph{The Michigan Mathematical Journal}}
  \bibinfo{volume}{32}, \bibinfo{number}{2} (\bibinfo{year}{1985}),
  \bibinfo{pages}{221--225}.
\newblock


\bibitem[\protect\citeauthoryear{Manber and Myers}{Manber and Myers}{1993}]%
        {manber1993suffix}
\bibfield{author}{\bibinfo{person}{Udi Manber} {and} \bibinfo{person}{Gene
  Myers}.} \bibinfo{year}{1993}\natexlab{}.
\newblock \showarticletitle{Suffix Arrays: A New Method for On-Line String
  Searches}.
\newblock \bibinfo{journal}{\emph{SIAM J. Comput.}} \bibinfo{volume}{22},
  \bibinfo{number}{5} (\bibinfo{year}{1993}), \bibinfo{pages}{935--948}.
\newblock


\bibitem[\protect\citeauthoryear{Munro and Raman}{Munro and Raman}{1996}]%
        {MUNRO1996311}
\bibfield{author}{\bibinfo{person}{J.~Ian Munro} {and}
  \bibinfo{person}{Venkatesh Raman}.} \bibinfo{year}{1996}\natexlab{}.
\newblock \showarticletitle{Selection from read-only memory and sorting with
  minimum data movement}.
\newblock \bibinfo{journal}{\emph{Theoretical Computer Science}}
  \bibinfo{volume}{165}, \bibinfo{number}{2} (\bibinfo{year}{1996}),
  \bibinfo{pages}{311 -- 323}.
\newblock
\showISSN{0304-3975}


\bibitem[\protect\citeauthoryear{Prezza}{Prezza}{2018}]%
        {prezzaSparse}
\bibfield{author}{\bibinfo{person}{Nicola Prezza}.}
  \bibinfo{year}{2018}\natexlab{}.
\newblock \showarticletitle{In-place Sparse Suffix Sorting}. In
  \bibinfo{booktitle}{\emph{Proceedings of the Twenty-Ninth Annual ACM-SIAM
  Symposium on Discrete Algorithms}} \emph{(\bibinfo{series}{SODA '18})}.
  \bibinfo{publisher}{Society for Industrial and Applied Mathematics},
  \bibinfo{address}{Philadelphia, PA, USA}, \bibinfo{pages}{1496--1508}.
\newblock
\showISBNx{978-1-6119-7503-1}


\bibitem[\protect\citeauthoryear{P\u{a}tra\c{s}cu and Viola}{P\u{a}tra\c{s}cu
  and Viola}{2010}]%
        {puatracscu2010cell}
\bibfield{author}{\bibinfo{person}{Mihai P\u{a}tra\c{s}cu} {and}
  \bibinfo{person}{Emanuele Viola}.} \bibinfo{year}{2010}\natexlab{}.
\newblock \showarticletitle{Cell-probe Lower Bounds for Succinct Partial Sums}.
  In \bibinfo{booktitle}{\emph{Proceedings of the Twenty-first Annual ACM-SIAM
  Symposium on Discrete Algorithms}} \emph{(\bibinfo{series}{SODA '10})}.
  \bibinfo{publisher}{Society for Industrial and Applied Mathematics},
  \bibinfo{address}{Philadelphia, PA, USA}, \bibinfo{pages}{117--122}.
\newblock
\showISBNx{978-0-898716-98-6}


\bibitem[\protect\citeauthoryear{Puglisi, Smyth, and Turpin}{Puglisi
  et~al\mbox{.}}{2007}]%
        {puglisi2007taxonomy}
\bibfield{author}{\bibinfo{person}{Simon~J. Puglisi}, \bibinfo{person}{W.~F.
  Smyth}, {and} \bibinfo{person}{Andrew~H. Turpin}.}
  \bibinfo{year}{2007}\natexlab{}.
\newblock \showarticletitle{A Taxonomy of Suffix Array Construction
  Algorithms}.
\newblock \bibinfo{journal}{\emph{ACM Comput. Surv.}} \bibinfo{volume}{39},
  \bibinfo{number}{2} (\bibinfo{date}{July} \bibinfo{year}{2007}),
  \bibinfo{pages}{4–es}.
\newblock
\showISSN{0360-0300}


\bibitem[\protect\citeauthoryear{Puglisi and Turpin}{Puglisi and
  Turpin}{2008}]%
        {puglisi2008space}
\bibfield{author}{\bibinfo{person}{Simon~J. Puglisi} {and}
  \bibinfo{person}{Andrew Turpin}.} \bibinfo{year}{2008}\natexlab{}.
\newblock \showarticletitle{Space-Time Tradeoffs for Longest-Common-Prefix
  Array Computation}. In \bibinfo{booktitle}{\emph{Algorithms and
  Computation}}, \bibfield{editor}{\bibinfo{person}{Seok-Hee Hong},
  \bibinfo{person}{Hiroshi Nagamochi}, {and} \bibinfo{person}{Takuro Fukunaga}}
  (Eds.). \bibinfo{publisher}{Springer Berlin Heidelberg},
  \bibinfo{address}{Berlin, Heidelberg}, \bibinfo{pages}{124--135}.
\newblock


\bibitem[\protect\citeauthoryear{Rabin}{Rabin}{1980}]%
        {rabin1980probabilistic}
\bibfield{author}{\bibinfo{person}{Michael~O Rabin}.}
  \bibinfo{year}{1980}\natexlab{}.
\newblock \showarticletitle{Probabilistic algorithm for testing primality}.
\newblock \bibinfo{journal}{\emph{Journal of number theory}}
  \bibinfo{volume}{12}, \bibinfo{number}{1} (\bibinfo{year}{1980}),
  \bibinfo{pages}{128--138}.
\newblock


\bibitem[\protect\citeauthoryear{Salowe and Steiger}{Salowe and
  Steiger}{1987}]%
        {SS1987}
\bibfield{author}{\bibinfo{person}{Jeffrey Salowe} {and}
  \bibinfo{person}{William Steiger}.} \bibinfo{year}{1987}\natexlab{}.
\newblock \showarticletitle{Simplified Stable Merging Tasks}.
\newblock \bibinfo{journal}{\emph{J. Algorithms}} \bibinfo{volume}{8},
  \bibinfo{number}{4} (\bibinfo{date}{Dec.} \bibinfo{year}{1987}),
  \bibinfo{pages}{557--571}.
\newblock
\showISSN{0196-6774}


\bibitem[\protect\citeauthoryear{Tanimura, Nishimoto, Bannai, Inenaga, and
  Takeda}{Tanimura et~al\mbox{.}}{2017}]%
        {tanimura17}
\bibfield{author}{\bibinfo{person}{Yuka Tanimura}, \bibinfo{person}{Takaaki
  Nishimoto}, \bibinfo{person}{Hideo Bannai}, \bibinfo{person}{Shunsuke
  Inenaga}, {and} \bibinfo{person}{Masayuki Takeda}.}
  \bibinfo{year}{2017}\natexlab{}.
\newblock \showarticletitle{{Small-Space LCE Data Structure with Constant-Time
  Queries}}. In \bibinfo{booktitle}{\emph{42nd International Symposium on
  Mathematical Foundations of Computer Science (MFCS 2017)}}
  \emph{(\bibinfo{series}{Leibniz International Proceedings in Informatics
  (LIPIcs)})}, \bibfield{editor}{\bibinfo{person}{Kim~G. Larsen},
  \bibinfo{person}{Hans~L. Bodlaender}, {and} \bibinfo{person}{Jean-Francois
  Raskin}} (Eds.), Vol.~\bibinfo{volume}{83}. \bibinfo{publisher}{Schloss
  Dagstuhl--Leibniz-Zentrum fuer Informatik}, \bibinfo{address}{Dagstuhl,
  Germany}, \bibinfo{pages}{10:1--10:15}.
\newblock
\showISBNx{978-3-95977-046-0}
\showISSN{1868-8969}


\end{thebibliography}

\end{document}